\newif\ifanon%
\anonfalse%
\newcommand{\anontext}[2]{\ifanon#1\else#2\fi}

\newif\ifarxivsubmit%
\arxivsubmittrue%
\newcommand{\arxivsubmittext}[2]{\ifarxivsubmit#1\else#2\fi}

\newif\ifcameraready%
\camerareadytrue%
\newcommand{\camerareadytext}[2]{\ifcameraready#1\else#2\fi}

\documentclass{llncs}
\usepackage[T1]{fontenc}

\usepackage[normalem]{ulem}
\usepackage[dvipsnames]{xcolor}

\usepackage{xargs}
\usepackage[colorinlistoftodos,prependcaption,textsize=small,disable]{todonotes}
\newcommandx{\fabian}[2][1=]{\todo[linecolor=gray,backgroundcolor=gray!25,bordercolor=gray,#1]{#2}}
\newcommandx{\lukas}[2][1=]{\todo[linecolor=blue,backgroundcolor=blue!25,bordercolor=blue,#1]{#2}}
\newcommandx{\needswork}[2][1=]{\todo[linecolor=red,backgroundcolor=red!25,bordercolor=red,#1]{#2}}
\newcommandx{\mohammad}[2][1=]{\todo[linecolor=Plum,backgroundcolor=Plum!25,bordercolor=Plum,#1]{#2}}
\newcommandx{\kevin}[2][1=]{\todo[linecolor=green,backgroundcolor=green!25,bordercolor=green,#1]{#2}}
\usepackage[bookmarksnumbered,unicode]{hyperref}
\hypersetup{%
  colorlinks=true, 
  urlcolor=blue, 
  linkcolor=blue, 
  citecolor=blue 
}

\usepackage{breakcites}
\usepackage{amsmath}
\usepackage{mathpartir}
\usepackage{amssymb}

\usepackage{amsthm}

\usepackage[inline]{enumitem}
\setlist[enumerate]*{label=(\arabic*)}
\setlist[itemize]{label=\textbullet}

\usepackage{esvect}
\renewcommand{\vec}[1]{\vv{#1}}

\usepackage{xspace}
\usepackage{stmaryrd}
\usepackage{mathtools}

\usepackage{cite}
\usepackage[capitalise]{cleveref}

\makeatletter
\newcommand*{\crefns}[1]{{\@cref@sortfalse\cref{#1}}}
\makeatother

\usepackage{nameref}
\makeatletter

\DeclareDocumentCommand \inferrule { s O {} m m o }{%
  \IfBooleanTF{#1}%
  {%
    \mpr@inferstar*[#2]{#3}{#4}%
  }{%
    \mpr@inferrule[#2]{#3}{#4}%
  }%
  \IfValueT{#5}%
  {%
    \my@name@inferrule{#5}%
  }%
}
\NewDocumentCommand \my@name@inferrule { m }{%
  \def\@currentlabelname{\textsc{#1}}%
}
\makeatother

\crefname{definition}{Def.}{Defs.}
\crefname{fact}{Fact}{Facts}
\crefname{lstlisting}{Listing}{Listings}

\usepackage{tasks}
\hypersetup{hypertexnames=false}
\newcommand{\labelparenths}[1]{(#1)}
\newcommand{\labelparenthstextup}[1]{\textup{\labelparenths{\textbf{#1}}}}
\newskip\listisep
\listisep\smallskipamount
\settasks{before-skip=\listisep,after-skip=\listisep,label-width=2em, label-align=right, label-format=\labelparenths, label=\arabic*}

\newtheorem*{theorem*}{Theorem}
\newtheorem*{lemma*}{Lemma}

\newtheorem{objective}{Objective}

\Crefname{assumption}{Assumption}{Assumptions}
\creflabelformat{assumption}{\textup{(#2#1#3)}}
\NewTasksEnvironment[
  counter=assumption,
  label=\arabic*,
  label-format=\labelparenthstextup,
  ref=\arabic*
]{assumes}[\item]

\usepackage[labelformat=simple]{subcaption}

\usepackage{environ}
\usepackage{upgreek}
\NewEnviron{code}{%
  \begingroup%
  \setlength{\jot}{0pt}%
  \begin{displaymath}%
    \BODY
  \end{displaymath}%
  \endgroup%
}
\NewEnviron{codealign}{%
  \begingroup%
  \setlength{\jot}{0pt}%
  \begin{flalign*}%
    \BODY
  \end{flalign*}%
  \endgroup%
}

\usepackage{adjustbox}

\usepackage{tikz}
\usepackage{algorithm}
\usepackage{algpseudocode}
\usepackage{wrapfig}

\usepackage{isabelle-listing}
\newcommand{\appendixref}[1]{%
\arxivsubmittext{\cref{#1}}%
{%
\camerareadytext{\protect\cite[\begin{NoHyper}\protect\hypersetup{hidelinks}\cref{#1}\end{NoHyper}]{todoarxivref}}{\cref{#1}}%
}}
\newcommand{\textiff}{if and only if\xspace}

\newcommand{\app}{\,}

\DeclarePairedDelimiter\parenths{(}{)}

\DeclarePairedDelimiter\abs{\lvert}{\rvert}
\newcommand{\define}{\equiv}

\newcommand{\constfont}[1]{\mathsf{#1}}
\newcommand{\mathtextfont}[1]{\mathrm{#1}}
\newcommand{\isabelle}{Isabelle\xspace}
\newcommand{\isabellehol}{\isabelle/HOL\xspace}

\newcommand{\mathcodefont}[1]{\mathtt{#1}}

\DeclarePairedDelimiterX\isatag[1]{\mathcodefont{[}}{\mathcodefont{]}}{#1}

\newcommand{\typefont}[1]{\mathsf{#1}}
\newcommand{\purefun}{\mathrel\typefont{\Rightarrow}}
\newcommand{\bool}{\typefont{bool}}
\newcommand{\valty}{\typefont{val}}
\newcommand{\stringty}{\typefont{string}}
\newcommand{\alphaty}{\typefont{\alpha}}
\newcommand{\betaty}{\typefont{\beta}}

\newcommand{\natty}{\typefont{\mathbb{N}}}

\newcommand{\relty}[2]{#1\purefun#2\purefun\bool}

\newcommand{\holimplies}{\longrightarrow}

\newcommand{\holhasty}{:}
\newcommand{\id}{\constfont{id}}
\newcommand{\map}{\constfont{map}}
\newcommand{\fold}{\constfont{fold}}
\newcommand{\comp}{\circ}
\newcommand{\vars}[1]{\constfont{regs}\app #1}
\newcommand{\ifconst}{\constfont{if}}
\newcommand{\thenconst}{\constfont{then}}
\newcommand{\elseconst}{\constfont{else}}
\newcommand{\holif}[3]{\ifconst\app #1\app\thenconst\app #2\app\elseconst\app #3}
\newcommand{\letconst}{\constfont{let}}
\newcommand{\inconst}{\constfont{in}}
\newcommand{\hollet}[3]{\letconst\app #1 = #2\app\inconst\app #3}
\DeclarePairedDelimiterXPP\funupd[3]{#1}{(}{)}{}{#2\coloneqq#3}

\newcommand{\stateretr}[2]{#1\app#2}

\newcommand{\natify}{\constfont{natify}}
\newcommand{\denatify}{\constfont{denatify}}
\newcommand{\pairnat}{\constfont{pair}}
\newcommand{\fstnat}{\constfont{fst}}
\newcommand{\sndnat}{\constfont{snd}}
\newcommand{\natselectorconst}{\constfont{select}}
\newcommand{\natselector}[2]{\natselectorconst_{#1,#2}}
\newcommand{\relnat}{\constfont{R}\natty}
\newcommand{\relnatify}[1]{#1^{\relnat}}
\newcommand{\caseconst}{\constfont{case}}
\newcommand{\casetype}[1]{\caseconst_{#1}}
\newcommand{\casenat}{\termholnat{\caseconst}}
\newcommand{\casetypenat}[1]{\termholnat{\caseconst_{#1}}}
\newcommand{\ofconst}{\constfont{of}}
\newcommand{\typnatify}[1]{#1^{\natty}}
\newcommand{\powsndnat}{\constfont{pow}\sndnat}

\newcommand{\tcabbrev}{TC} 
\newcommand{\cabbrev}{C} 
\newcommand{\wabbrev}{W} 
\newcommand{\minus}{-} 

\newcommand{\hol}{\mathtextfont{HOL}}
\newcommand{\termhol}[1]{#1}
\newcommand{\natabbrev}{\natty}
\newcommand{\holnat}{\hol^{\natabbrev}}

\newcommand{\termholnat}[1]{#1^{\natabbrev}}
\newcommand{\termfolnat}[1]{#1^{\natabbrev}}
\newcommand{\holtc}{\hol^{\mathtextfont{\tcabbrev}}}
\newcommand{\foltc}{\hol^{\mathtextfont{\tcabbrev}}}
\newcommand{\termholtc}[1]{#1^{\mathtextfont{\tcabbrev}}}

\newcommand{\tcnatabbrev}{\mathtextfont{\tcabbrev}\natty}
\newcommand{\holtcnat}{\hol^{\tcnatabbrev}}
\newcommand{\foltcnat}{\hol^{\tcnatabbrev}}
\newcommand{\termholtcnat}[1]{#1^{\tcnatabbrev}}
\newcommand{\termfoltcnat}[1]{#1^{\tcnatabbrev}}

\newcommand{\imp}{\typefont{IMP}}

\newcommand{\imptc}{\imp^{\mathtextfont{\tcabbrev}}}
\newcommand{\impwc}{\imp^{\mathtextfont{\wabbrev\cabbrev}}}
\newcommand{\impc}{\impwc}
\newcommand{\impw}{\imp^{\mathtextfont{\wabbrev}}}
\newcommand{\impminus}{\imp^{\mathtextfont{\minus}}}

\newcommand{\const}{\mathcal{C}}
\newcommand{\psubst}[2]{[#1 / #2]}
\newcommand{\impconstrfont}[1]{\mathtt{#1}}

\DeclarePairedDelimiterXPP\aval[2]{}{\llbracket}{\rrbracket}{_{#2}}{#1}

\newcommand{\aconstop}{\impconstrfont{C}}
\newcommand{\aconst}[1]{\aconstop\app#1}
\newcommand{\avarop}{\impconstrfont{R}}
\newcommand{\avar}[1]{\avarop\app#1}

\newcommand{\bigsteparrow}{\Rightarrow}

\newcommand{\imptccontextdelim}{\vdash}
\newcommand{\bigstep}[3]{(#1,#2)\bigsteparrow#3}

\newcommand{\bigstepDelim}[3]{\imptccontextdelim\bigstep{#1}{#2}{#3}}
\newcommand{\bigstepCtxt}[4]{#1\bigstepDelim{#2}{#3}{#4}}
\newcommand{\bigstepTime}[4]{(#1,#2)\bigsteparrow^{#3}#4}

\newcommand{\bigstepPart}[4]{(#1,#2)\bigsteparrow_{#3}#4}
\newcommand{\bigstepDelimTime}[5]{\imptccontextdelim\bigstepTime{#1}{#2}{#3}{(#4,#5)}}

\newcommand{\bigstepCtxtTime}[5]{#1\imptccontextdelim\bigstepTime{#2}{#3}{#4}{#5}}
\newcommand{\bigstepCtxtPart}[5]{#1\imptccontextdelim\bigstepPart{#2}{#3}{#4}{#5}}
\newcommand{\bigstepTimePart}[5]{(#1,#2)\bigsteparrow^{#3}_{#4}#5}
\newcommand{\bigstepDelimTimePart}[5]{\imptccontextdelim\bigstepTimePart{#1}{#2}{#3}{#4}{#5}}
\newcommand{\bigstepCtxtTimePart}[6]{#1\bigstepDelimTimePart{#2}{#3}{#4}{#5}{#6}}

\newcommand{\imptctoimpc}[1]{\llparenthesis #1 \rrparenthesis_{\circlearrowleft}}
\newcommand{\impctoimpw}[1]{\llparenthesis #1 \rrparenthesis_{\star}}

\newcommand{\copyopnd}[3]{\constfont{copy}^{#3}\app#1\app#2}
\newcommand{\adder}[2]{\constfont{adder}^{#2}\app#1}
\newcommand{\fulladder}[2]{\constfont{fullAdder}^{#2}\app#1}

\newcommand{\bitblast}[2]{{\llparenthesis #1 \rrparenthesis}^{#2}}
\newcommand{\bitblastAexp}[3]{{\llparenthesis #1 \rrparenthesis}^{#2}_{#3}}
\newcommand{\bitblastState}[2]{{\llparenthesis #1 \rrparenthesis}^{#2}}
\newcommand{\impminusminusadd}[4]{\impconstrfont{add}^{#4}\app#1\app#2\app#3}

\newcommand{\assignN}{Assign}
\newcommand{\whileTrueN}{WhT}
\newcommand{\whileFalseN}{WhF}
\newcommand{\ifTrueN}{IfT}
\newcommand{\ifFalseN}{IfF}
\newcommand{\seqN}{Seq}
\newcommand{\tailN}{Rec}
\newcommand{\callN}{Call}

\newcommand{\assignop}{\mathrel{\impconstrfont{\leftarrow}}}
\newcommand{\assign}[2]{#1\assignop#2}
\newcommand{\seqop}{\mathrel{\impconstrfont{;}}}
\newcommand{\seq}[2]{#1\seqop#2}
\DeclareMathOperator{\seqi}{\seqop}
\newcommand{\iif}[3]{\impconstrfont{IF}\ #1\ \impconstrfont{THEN}\ #2\ \impconstrfont{ELSE}\ #3}
\newcommand{\tail}{\impconstrfont{RECURSE}}
\newcommand{\callop}{\impconstrfont{CALL}}
\newcommand{\call}[2]{\callop\ #1\ \impconstrfont{RETURN}\ #2}
\newcommand{\whileop}{\impconstrfont{WHILE}}
\newcommand{\while}[2]{\whileop\ #1\ \impconstrfont{DO}\ #2}

\DeclareMathOperator{\statenormop}{\leadsto}
\newcommand{\statenorm}[2]{#1\statenormop#2}

\newcommand{\relarrow}{\Rrightarrow}
\newcommand{\funrel}[2]{#1\relarrow#2}

\newcommand{\depfunmapcolons}{::}
\newcommand{\funmaparrow}{\rightarrow}
\DeclarePairedDelimiterXPP\depfunmap[3]{}{[}{]}{\funmaparrow#3}{#1\depfunmapcolons#2}

\newcommand{\htoiargconst}{\constfont{arg}}
\newcommand{\htoiargi}[2]{\htoiargconst_{#1,#2}}

\newcommand{\htoiretconst}{\constfont{ret}}
\newcommand{\htoiret}[1]{\htoiretconst_{#1}}

\newcommand{\maxconst}[1]{\ensuremath{#1_{\text{max}}}}

\newcommand{\holnIfC}{\textbf{If}}
\newcommand{\holnLetC}{\textbf{Let}}
\newcommand{\holnLetBoundC}{\textbf{LetBound}}
\newcommand{\holnArgC}{\textbf{Arg}}
\newcommand{\holnCallC}{\textbf{Call}}
\newcommand{\holnTailCallC}{\textbf{Recurse}}
\newcommand{\holnNumC}{\textbf{Number}}

\newcommand{\holnIf}[3]{\holnIfC{} \app #1 \app #2 \app #3}
\newcommand{\holnLet}[2]{\holnLetC{} \app #1 \app #2}
\newcommand{\holnLetBound}[1]{\holnLetBoundC{} \app #1}
\newcommand{\holnArg}[1]{\holnArgC{} \app #1}
\newcommand{\holnCall}[2]{\holnCallC{} \app #1 \app [#2]}
\newcommand{\holnTailCall}[1]{\holnTailCallC{} \app [#1]}
\newcommand{\holnNum}[1]{\holnNumC{} \app #1}

\newcommand{\dens}[3]{{\llbracket #1 \rrbracket}^{#2}_{#3}}

\newcommand{\maxset}[1]{\ensuremath{\text{max}\app#1}}
\newcommand{\transport}{\textsc{Transport}\xspace}
\newcommand{\whiteboxtransport}[1]{#1^{%
\begin{tikzpicture} [dot/.style={draw,dash pattern=on 1.5pt off 1.2pt,rectangle,minimum size=2mm,inner sep=0pt,outer sep=0pt}]\path (0,0) coordinate[dot];\end{tikzpicture}}}

\newcommand{\threesat}{\textsc{3SAT}\xspace}
\newcommand{\sat}{\textsc{SAT}\xspace}
\newcommand{\true}{\impconstrfont{1}}
\newcommand{\false}{\impconstrfont{0}}

\begin{document}
\arxivsubmittext{\pagestyle{plain}}{\pagestyle{plain}}

\title{Proof-Producing Translation of Functional Programs into a Time \& Space Reasonable Model}

\author{\anontext{Anon.\ Author}{Kevin Kappelmann, Fabian Huch\camerareadytext{\thanks{Funded by the Deutsche Forschungsgemeinschaft (DFG, German Research Foundation) under the National Research Data Infrastructure – NFDI 52/1 – 501930651}}{}, Lukas Stevens, Mohammad Abdulaziz}}
\institute{\anontext{Anon.\ Place\\\email{anon.\ email}}{%
King's College London, Bush House 30 Aldwych, London WC2B 4BG, UK,\\
\email{mohammad.abdulaziz@kcl.ac.uk}\\
Technische Universität München, Boltzmannstrasse 3, Garching 85748, Germany,\\
\email{kevin.kappelmann@tum.de}, \email{huch@in.tum.de}, \email{lukas.stevens@in.tum.de}}}
\maketitle              
\begin{abstract}
We present a semi-automated framework
to construct and reason about programs in a deeply-embedded while-language.
The while-language we consider is a simple computation model that can simulate (and be simulated by) Turing Machines with a quadratic time and constant space blow-up.
Our framework derives while-programs from functional programs written in a subset of \isabellehol,
namely tail-recursive functions with first-order arguments and algebraic datatypes.
As far as we are aware,
it is the first framework targeting a computation model
that is reasonable in time and space from a complexity-theoretic perspective.
\end{abstract}

\section{Introduction}

A large array of mathematical results has been formlised in interactive theorem provers.
One example is theoretical computer science, where extensive results on algorithm correctness
and computational objects, such as reductions, were proven.
Nonetheless, one type of results has proved to be particularly hard to formalise:
the construction and verification of algorithms modelled within a deeply embedded computation model, such as a Turing machine.
These constructions, however, are crucial to formalise the theory of efficient algorithms and computational complexity theory
as understood by practitioners of theoretical computing.
In these fields, most results show that an algorithm runs in polynomial time (worst-case, average case, etc.) when implemented as a program within a reasonable computation model, i.e.\ a model that can simulate Turing machines with a polynomial time and constant space overhead for \emph{all} computations~\cite{boasReasonableModels}.
Previous work on reasoning about concrete programs within a deeply embedded computation model did so in one of two ways:
\begin{enumerate}
\item Automated synthesis of deeply embedded programs
from functional programs in the theorem prover's logic
along with equivalence proofs~\cite{coqexctractionlambda,ProofProducingML}.
In the context of formalising complexity theory,
this approach was only used for a lambda calculus model so far~\cite{forsterThesis}.
In practical verification contexts,
it was used for a lambda calculus model~\cite{ProofProducingML}
and assembly~\cite{coqassemblysynthesis}.
\item Interactively proving properties of deeply embedded programs, mainly using Hoare or other program logics.
This approach is used in many practical verification settings~\cite{AutoCorres,Iris,WhileHOLfour}
and has been applied to imperative languages, including toy languages~\cite{ConcreteSemantics}, C, Java, and Rust.
In the context of formalising theoretical computer science, this approach was mainly used for Turing machines~\cite{urbanTuringMachines,turingMachinesMatita,forsterTuringMachines},
with few applications to lambda calculus~\cite{norrishComputabilityThy,weaklambdaembeddedcoq} and partial recursive functions~\cite{carneiropartrec}.
\end{enumerate}
To formalise the theory of efficient algorithms and complexity theory,
we suggest the synthesis into a simple, imperative, and reasonable computation model.

\paragraph{Why?}
As others have noted~\cite{CookLevinCoq,coqexctractionlambda}, reasoning about programs in a theorem prover’s native language is significantly easier than reasoning about deeply embedded programs.
This is crucial when verifying complex algorithms.
For example, G{\"a}her and Kunze~\cite{CookLevinCoq}
proved the Cook-Levin theorem with native Coq programs,
benefiting from existing tools and automation.
In contrast, Balbach~\cite{CookLevinTMIsabelle}
directly reasoned about Turing machines in \isabellehol, requiring extensive low-level reasoning.
Others~\cite{IMP2} note difficulties when reasoning about deeply-embedded programs, e.g. missing warning messages and no static type checking.


Although automated synthesis methods exist for lambda calculi, we believe it is worthwhile to develop them for an imperative model.
The main reason is that the lambda calculi used in previous work
are \emph{less equivalent} to Turing machines than imperative models:
they are only reasonable for computations solving decision problems but not reasonable in space in general~\cite{callbyvaluelambdareasonable}.
This limits their applicability when analysing complexity-theoretic questions.

Finally, while there exists a synthesis-based method for x86-assembly in Coq~\cite{coqassemblysynthesis},
it does not support user-defined datatypes, only allows pattern matching and recursion with fixed recursion schemes,
and its target language's semantics are vastly more complicated than needed for complexity-theoretic questions.

\paragraph{Contribution} We present a framework that
synthesises verified programs in a simple while-language that is reasonable ($\impw$) w.r.t.\ time and space in \isabellehol.
Our framework takes tail-recursive, first-order functions\footnote{%
Although the framework does not support higher-order functions,
each first-order instance of such a function can be compiled as described in \cref{sec:holtcnattoimptc}.} with algebraic datatypes as input.
We call this fragment $\foltc$.
Given a $\foltc$ function $f$,
our proof-producing metaprogram
synthesises a related $\foltc$ function which only operates on natural numbers.
We call this $\foltc$ fragment~$\foltcnat$.

The resulting $\foltcnat$ function $\termfoltcnat{f}$ is
compiled into a program $p$ of our deeply embedded, imperative language $\imptc$,
which supports function calls and tail-recursion.
We provide automation to interactively prove the correspondence between $\termfoltcnat{f}$ and $p$.
In all examples we tried, the proofs were automatic.
Finally, a verified chain of compilers in \isabellehol translates $\imptc$ programs into programs in our target computation model $\impw$.
\cref{fig:pipeline} gives an overview.
The blow-up introduced by each compiler is linear in the execution time.

Furthermore, we verify a compiler for $\impw$ programs that compiles them into programs in $\impminus$, a computation model with fixed bit-width registers.
This is a step towards showing that our framework and computation model are suited for formalising significant theorems in complexity theory, like the Cook-Levin Theorem, many of which need bit-blasting into a bounded bit-width computation model as a major proof step.
We show that this compilation step is correct and leads to a running blow-up linear in the product of the $\impw$ program running time and the size of the input.
\begin{figure}[t]
    \centering
    \begin{tikzpicture}
    \node (FOLTC) {$\foltc$};
    \node (FOLTCN) [right=5.5cm of FOLTC] {$\foltcnat$};
    \node (IMPTC) [below=0.5cm of FOLTCN] {$\imptc$};
    \node (IMPw) [below right=0.5 cm and 2cm of FOLTC] {$\impw$};
    \node (IMPminus) [below=0.5cm of FOLTC] {$\impminus$};

    \draw[->] (FOLTC) -- (FOLTCN) node[midway, above] {proof-producing metaprogram};
    \draw[->] (FOLTCN) -- (IMPTC) node[midway, right] {interactive deep embedding};
    \draw[->] (IMPTC) -- (IMPw) node[midway, above] {verified compilers};
    \draw[->] (IMPw) -- (IMPminus) node[midway, above] {verified compilers};

    \end{tikzpicture}
    \caption{Compilation pipeline from $\foltc$ functions to $\impminus$ programs.}
    \label{fig:pipeline}
\end{figure}

The target language only has one-bit-wide registers and only allows comparison to zero as well as bit assignments.
Both operations take constant time.
This is in addition to the regular program constructs of if-then-else and while-loops.
Hence it is clear that $\impminus$ is reasonable with respect to memory and time.


\paragraph{Availability} This article’s
\arxivsubmittext{supplementary material\footnote{\href{https://archive.softwareheritage.org/swh:1:dir:0885d0653246b8b2e81f168ed4a596d29200deed}{swh:1:dir:0885d0653246b8b2e81f168ed4a596d29200deed}}}{\camerareadytext{extended version~\cite{todoarxivref}}{supplementary material\footnote{\href{https://archive.softwareheritage.org/swh:1:dir:0885d0653246b8b2e81f168ed4a596d29200deed}{swh:1:dir:0885d0653246b8b2e81f168ed4a596d29200deed}}}}
\needswork{add the right URL}
includes the formalisation and a guide linking all definitions, results, and examples to their counterpart in the formalisation.
\cref{thm:spaceConsume,thm:comp_t_correct,thm:comp_c_correct,lem:bitblastProgram}
and all underlying sub-results are formalised in \isabellehol.
Other proofs are provided in the appendix.

\section{Preliminaries}

We built our framework in \isabellehol~\cite{isabellehol},
but our approach could be followed in any simple type theory (higher-order logic)~\cite{churchstt} or more expressive foundation.
We use standard lambda calculi syntax, i.e.
$t \Coloneqq  \lambda x.\app t \mid t_1\app t_2 \mid x \mid f$,
where $\lambda x.\app t$ denotes function abstractions,
$t_1\app t_2$ function applications,
$x$ bound variables, and
$f$ functions with defining equations $f=\lambda\vec{x}.\app t$.
We write $t\holhasty\alphaty$ for ''$t$ has type $\alphaty$''.
We use vector notation for $n$-ary applications, abstractions, and inputs:
\begin{align}
  f\app\vec{x}&\define f\app x_1\dotsb x_n\quad&
  \lambda\vec{x}.\app t&\define\lambda x_1\dotsb x_n.\app t\quad\\
  \vec{f_i\app x_i}&\define (f_1\app x_1)\dotsb (f_n\app x_n)\quad&
  \vec{\alphaty}\purefun\alphaty&\define \alphaty_1\purefun\dotsb\alphaty_n\purefun\alphaty
\end{align}
A function $f\holhasty\vec{\alpha}\purefun\alpha$ is
called \emph{higher-order} if some $\alphaty_i$ is a function type.
Otherwise, it is called \emph{first-order}.
An application $f\app\vec{x}\holhasty\alphaty$ is called \emph{higher-order}
if $\alphaty$ is a function type.
Otherwise, it is called \emph{first-order}.
An \emph{algebraic datatype} (ADT) $\vec{\alphaty}\app d=C_1\app \vec{\alphaty_1}\mid\dotsb\mid C_n\app\vec{\alphaty_n}$
is a possibly recursive sum type of product types
with constructors $C_i\holhasty \vec{\alphaty_i}\purefun \vec{\alphaty}\app d$
and case combinator
$\casetype{d}\holhasty \parenths{\vec{\alphaty_1}\purefun\alphaty}\purefun{\dotsb\purefun\parenths{\vec{\alphaty_n}\purefun\alphaty}\purefun\vec{\alphaty}\app d\purefun\alphaty}$.


Our framework takes as an \emph{input language} tail-recursive, first-order functions $f=\lambda\vec{x}.\app t$,
where $\vec{x}$ are the \emph{arguments} bound in the \emph{body} $t$.
The body may use other such functions and ADT functions in first-order applications.
\cref{fig:holadts} shows the grammar for recursive, first-order function bodies.
Our input language $\foltc$ is the subset where each recursive call $f\app\vec{t}$ is in tail-position.
\begin{figure}[t]
\begin{gather*}
\begin{aligned}
t \Coloneqq  & \; \hollet{x}{t_1}{t_2} & \text{(bind $t_1$ to $x$ in $t_2$)} \\
  \mid & \; (\casetype{d}\app t\app\ofconst\app \vec{x_1}\Rightarrow t_1\mid\dotsb\mid\vec{x_n}\Rightarrow t_n)  & \text{(case combinator of ADT $\vec{\alphaty}\app d$)}\\
  \mid & \; h\app \vec{t} & \text{(first-order application)}\\
  h \Coloneqq  & \; x  \mid g \mid f & \text{(bound variable $x$, function $g$, or $f$ itself)}
\end{aligned}
\end{gather*}
\caption{Grammar for body $t$ of a first-order function $f=\lambda\protect\vec{x}.\app t$ with ADTs.}\label{fig:holadts}
\end{figure}

We use several deeply-embedded, deterministic languages based on the \emph{imperative language} $\imp$ from Winskel's book~\cite{winskelOpSem},
which is well-studied in formal verification~\cite{coqSemanticsIMP,nipkowSemantics}.
Our $\imp$ languages operate on a \emph{state} of type
$\stringty\purefun \valty$.
Terms of type $\stringty$ are called \emph{registers}
and terms of type $\valty$ are called \emph{values}.
The registers in a program $p$ are denoted by $\vars\app p$.
As a computation model, $\imp$ allows simple reasoning about space since there is no
indirect memory addressing --
only registers occurring in a program can be accessed.
In our final model, $\impminus$, values are single bits,
making the space usage bounded by $\vars\app p$.
We define $\impminus$ in \cref{sec:impwtoimpminus},
and a summary of all $\imp$ languages can be found in~\appendixref{sec:all_imps}.
Here, we describe the intermediate languages used during compilation.
They use values of type $\natty$ and
consist of \emph{atoms}, \emph{(arithmetic) expressions} and \emph{(imperative) commands}.
Every atom is also an arithmetic expression.

We now define the \emph{semantics}.
With $\aval{A}{s}$ and $\aval{a}{s}$ we denote the evaluation (returning a natural number) of an
atom $A$ and expression $a$, respectively, under state $s$,
defined by the rules in \cref{fig:aval_atom_imp,fig:aval_imp}.
\begin{figure}[t]
\begin{subfigure}[t]{0.4\textwidth}
\begin{gather*}
\aval{\aconst{n}}{s}\define n,
\hspace{0.9em}
\aval{\avar{r}}{s}\define \stateretr{s}{r}
\end{gather*}
\caption{Evaluation of atoms $A$.}\label{fig:aval_atom_imp}
\end{subfigure}
\begin{subfigure}[t]{0.6\textwidth}
\begin{gather*}
\aval{A_1\otimes A_2}{s}\define \aval{A_1}{s} \otimes \aval{A_2}{s},
\text{ for } \otimes\in\{+,-\}
\end{gather*}
\caption{Evaluation rules for arithmetic expressions.}\label{fig:aval_imp}
\end{subfigure}
\begin{subfigure}[t]{\textwidth}
\begin{mathpar}
  \inferrule*[left=\assignN]
  {s'=\funupd{s}{r}{\aval{a}{s}}}
  {\bigstepTime{\assign{r}{a}}{s}{\const}{s'}}[\assignN]\label{infrule:assign}

  \inferrule*[left=\ifTrueN]
  {\stateretr{s}{r} \neq 0 \\ \bigstepTime{p_1}{s}{n}{s'}}
  {\bigstepTime{\iif{r}{p_1}{p_2}}{s}{n+\const}{s'}}[\ifTrueN]\label{infrule:iftrue}

  \inferrule*[left=\seqN]
  {\bigstepTime{p_1}{s}{n_1}{s'} \\
  \bigstepTime{p_2}{s'}{n_2}{s''}}
  {\bigstepTime{\seq{p_1}{p_2}}{s}{n_1+n_2+\const}{s''}}[\seqN]\label{infrule:seq}

  \inferrule*[left=\ifFalseN]
  {\stateretr{s}{r} = 0 \\ \bigstepTime{p_2}{s}{n}{s'}}
  {\bigstepTime{\iif{r}{p_1}{p_2}}{s}{n+\const}{s'}}[\ifFalseN]\label{infrule:iffalse}
\end{mathpar}
\caption{Execution relation for commands,
where each $\const$ is a fixed, non-negative constant.
}\label{fig:exec_imp}
\end{subfigure}
\caption{Semantics shared by our $\imp$-languages, excluding the final $\impminus$.}\label{fig:imp_semantics}
\end{figure}
With
$\bigstepCtxtTimePart{tp}{p}{s}{n}{R}{s'}$
we denote that in context of the program $tp$,\footnote{The context is only relevant for recursive calls in \cref{sec:holtcnattoimptc}.}
the program $p$ started from $s$ terminates in at most $n$ steps with state
$s''$ such that $\stateretr{s'}{r}=\stateretr{s''}{r}$ for all $r\in R$:
\begin{gather*}
\parenths[\big]{\bigstepCtxtTimePart{tp}{p}{s}{n}{R}{s'}} \define \exists n'\, s''.\; \bigstepCtxtTime{tp}{p}{s}{n'}{s''} \land \forall r\in R.\app \stateretr{s'}{r}=\stateretr{s''}{r} \land n'\leq n,
\end{gather*}
where $\bigsteparrow^n$ is the \emph{(big-step) execution relation},
defined inductively using the rules from \cref{fig:exec_imp} along with the language-specific rules specified in their respective sections.
In rule~\nameref{infrule:assign},
the state $s$ is updated at register $r$
using the state update notation $\funupd{s}{r}{v}\define\lambda r'.\app \holif{r=r'}{v}{s\app r'}$.
If we are only interested in a single return register $r$, we write
\begin{gather*}
\parenths[\big]{\bigstepCtxtPart{tp}{p}{s}{r}{v}}\define \exists s'.\; \parenths[\big]{\bigstepCtxt{tp}{p}{s}{s'}\land \stateretr{s'}{r}=v}.
\end{gather*}
\todo{Example of while-loop and reasoning about it...}
To simplify notation, we
omit ``$\dots \imptccontextdelim$'' if the context is irrelevant,
$R$ if the program terminates exactly with state $s'$,
and $n$ if the number of steps is irrelevant.

One may question why rule~\nameref{infrule:assign} only takes constant, and not logarithmic, time.
This is merely for simplicity and makes no difference for polynomial time considerations since
logarithmic space usage (i.e.\ space usage in a binary encoding of $\natty$)
is bounded by time in all intermediate $\imp$ languages, as shown in \cref{thm:spaceConsume}.
For this, let $\maxconst{a}$ denote the largest constant in an expression $a$,
$\maxconst{p}$ the largest constant in a program~$p$,
and $\maxconst{s}\define\maxset \{\stateretr{s}{r} \mid r\holhasty\stringty\}$.
\begin{theorem}\label{thm:spaceConsume}
If $\maxset\{\maxconst{s},\maxconst{p}\} < 2^w$ and
$\bigstepTime{p}{s}{n}{s'}$,
then $\maxconst{s'} < 2^{w+n}$.
\end{theorem}
\begin{proof}[Proof sketch]
The proof is by induction over the execution relation.
In case of an assignment, the maximum size can at most double.
In other cases, the claim holds by applying the induction hypothesis, followed by algebraic manipulation.
\end{proof}

\section{$\hol^{(\mathtextfont{\tcabbrev})}$ to $\hol^{(\mathtextfont{\tcabbrev})\natty}$}\label{sec:holtctoholtcnat}
This section describes
the proof-producing translation of $\foltc$ functions into
equivalent $\foltcnat$ functions.
More generally, our translation applies to any $\hol$ function $f$ with ADTs,
producing a $\holnat$ function $\termholnat{f}$
that operates only on natural numbers
while preserving tail-recursiveness and first-order applications.

Intuitively, the function $f$ and its translation $\termholnat{f}$ should compute the same values for the same inputs,
up to encoding and decoding.
Our method is based on a technique called \emph{transport}, which, intuitively, takes
objects of one type and derives that ``corresponding'' objects exist of another type,
synthesising the latter objects as necessary.
To formalise this correspondence, we first review some concepts on relations.
\begin{definition}\label{def:relation}
A \emph{relation on $\alphaty$ and $\betaty$} is a function of type
$\relty{\alphaty}{\betaty}$.
Two functions $f,g$ are \emph{related from $R$ to $S$} if they map $R$-related inputs to $S$-related outputs,
written $\parenths{\funrel{R}{S}}\app f\app g\define \forall x\app y.\app R\app x\app y\holimplies S\app(f\app x)\app(g\app y)$.
Repeated relations are abbreviated as
$\parenths{\funrel{R^n}{S}}\define{}\parenths{{R\relarrow\dotsb\relarrow R}\relarrow S}$, i.e. $R$ repeated $n$-times.
\end{definition}

\let\oldholtc\holtc
\let\oldholtcnat\holtcnat
\let\oldtermholtc\termholtc
\let\oldtermholtcnat\termholtcnat
\renewcommand{\holtc}{\hol}
\renewcommand{\holtcnat}{\holnat}
\renewcommand{\termholtc}{\termhol}
\renewcommand{\termholtcnat}{\termholnat}

We can now define the desired relation between $\termholtc{f}$ and $\termholtcnat{f}$.
\begin{definition}\label{def:nat_encoded}
A type $\alphaty$ is \emph{encodable} if there are functions
$\natify\holhasty \alpha\purefun\natty$
and $\denatify\holhasty \natty\purefun \alphaty$
such that
$\denatify\app(\natify\app x)=x$ for all $x$.\footnote{In the formalisation, $\natify,\denatify$ are overloaded as part of a typeclass.}

We say that $n\holhasty\natty$ and $x\holhasty\alphaty$ are \emph{($\relnat$-)related}
if they represent the same value:
$\relnat\app n\app x \define n = \natify\app x$.
A number $n$ is \emph{well-encoded (with respect to~$\alphaty$)}
if $\relnat\app n\app x$ holds for some $x\holhasty\alphaty$.
We lift this relation to functions
and say that $f^\natty\holhasty \typnatify{\alphaty}$ and $f\holhasty\alphaty$ are \emph{$\relnat$-related}
if $\relnatify{\alphaty}\app\termholtcnat{f}\app\termholtc{f}$,
where $\typnatify{\alpha}, \relnatify{\alpha}$ are defined recursively:
\begin{gather*}
\typnatify{(\alphaty \purefun \betaty)} \define \typnatify{\alphaty} \purefun \typnatify{\betaty},\quad
\typnatify{\alphaty}  \define \natty,\qquad
\relnatify{(\alphaty \purefun \betaty)} \define \funrel{\relnatify{\alphaty}}{\relnatify{\betaty}},\quad
\relnatify{\alphaty}  \define \relnat.
\end{gather*}

\end{definition}
\begin{objective}\label{obj:holtctoholtcnat}
Given a $\holtc$ function $\termholtc{f}$,
synthesise a $\relnat$-related $\holtcnat$ term~$\termholtcnat{f}$.
Moreover, the synthesis preserves tail-recursiveness and first-order applications on well-encoded inputs.
\end{objective}
We achieve this in two steps:
First, we encode the datatypes used by $\termholtc{f}$ (\cref{sec:natify_datatypes}).
Second, we synthesise a $\holtcnat$ term $\termholtcnat{f}$ related to $\termholtc{f}$ (\cref{sec:natify_funs}).

\subsection{Encoding of Datatypes\label{sec:natify_datatypes}}

Let $\vec{\alphaty}\app d=C_1\app \vec{\alphaty_1}\mid\dotsb\mid C_n\app\vec{\alphaty_n}$ be an ADT
with $a_i\define \abs{\vec{\alphaty_i}}$.
We want to encode $\vec{\alphaty}\app d$ according to \cref{def:nat_encoded}.
Assuming that all $\alpha\in \vec{\alpha_1}\cup\dotsb\cup\vec{\alpha_n}$
are encoded or \makebox{$\alpha=\vec{\alpha}\app d$},
the existence of suitable en- and decoding functions is straightforward.
Following Gordon's approach~\cite{gordonHOLTypes} to encode the type of lists,
each value $C_i\app \vec{x}$ is encoded as a tagged and nested pair.
Datatypes are encoded recursively this way, using natural numbers as the base type:
\begin{definition}\label{def:natify_constructor}
Fix an injective pairing function $\pairnat\holhasty\natty^2\purefun\natty$
with inverses $\fstnat, \sndnat$, i.e.\
$\fstnat\app(\pairnat\app n\app m)=n$
and
$\sndnat\app(\pairnat\app n\app m)=m$.
We define \makebox{$\termholnat{C_i}\holhasty \natty^{a_i}\purefun \natty$} by
$\termholnat{C_i}\app\vec{x}\define
\pairnat\app i\app\parenths[\big]{
\holif{a_i=0}
{0}
{\parenths{\pairnat\app x_1\app (\pairnat\app x_2(\dotsb(\pairnat\app x_{a_i-1}\app x_{a_i})\dotsb}}}$.
For each $i,j$, we define the \emph{selector}
$\natselector{i}{j}\app x\define \parenths{\holif{j<i}{\fstnat}{\id}}\app(\sndnat^j\app x)$.
Using the selector,
we can define a companion $\casetypenat{d}$
for the case combinator $\casetype{d}$
using a nested if-then-else construction.
Details can be found in~\appendixref{sec:natify_datatypes_details}.
\end{definition}
These functions suffice to show that $\vec{\alphaty}\app t$ is encodable:
\begin{theorem}\label{thm:nat_encoded_datatype}
$\vec{\alphaty}\app d=C_1\app \vec{\alphaty_1}\mid\dotsb\mid C_n\app\vec{\alphaty_n}$ is encodable
if all $\alpha\in \vec{\alpha_1}\cup\dotsb\cup\vec{\alpha_n}$ are encoded or \makebox{$\alpha=\vec{\alpha}\app d$}.
Also, for all $1\leq i\leq n$, $1\leq j\leq a_i$ and relations $R$, we have
\begin{gather*}
\parenths{\funrel{\relnat^{a_i}}{\relnat}}\app \termholnat{C_i}\app C_i,\qquad\qquad
\app\relnat\app n\app(C_i\app \vec{x})\implies \relnat\app (\natselector{a_i}{j}{n})\app x_j,\\
\parenths[\big]{\funrel{\parenths{\funrel{\relnat^{a_1}}{R}}}{\funrel{\dotsb\relarrow\parenths{\funrel{\relnat^{a_n}}{R}}}{\funrel{\relnat}{R}}}}\app \casetypenat{d}\app \casetype{d}.
\end{gather*}
\end{theorem}
The described pipeline is fully automatic in our framework.
An example application can be found in~\appendixref{sec:natify_datatypes_details}.
We note that our underlying pairing function is the Cantor pairing function, which results in a number computable in time polynomial in the sizes of the two numbers in the encoded pair.

\subsection{Synthesis of $\holtcnat$ Functions}\label{sec:natify_funs}

Let $\termholtc{f}=\lambda\vec{x}.\app \termholtc{t}$
be a $\holtc$ function of type $\alphaty$.
We aim for \cref{obj:holtctoholtcnat}.
For this, we can assume that all ADTs in $\alpha$ are encoded
and every user-specified function $\termholtc{g}\neq \termholtc{f}$ in $\termholtc{t}$
is already compiled -- and thus $\relnat$-related to some $\termholtcnat{g}$.

Our approach combines \emph{black-box and white-box transports}~\cite{transport,transportcoq,huffmanTransfer}
in a novel way, to the best of our knowledge.
Both kinds of transport take a term $t\holhasty\alphaty$ and a relation $R$ as input
and synthesise a term $t'$ together with a proof that $R\app t'\app t$.
While black-box transports do so using only the structure of $\alphaty$ and $R$,
white-box transports also use the structure of $t$.

Our approach is as follows:
\begin{enumerate*}
\item Use black-box transport to obtain a function $\termholtcnat{f}$ that is $\relnat$-related to $\termholtc{f}$.
\item Use white-box transport to obtain a term $\whiteboxtransport{(\lambda\vec{x}.\app \termholtc{t})}$ that uses $\termholtcnat{f}$
  and which is also $\relnat$-related to $\termholtc{f}$.
\item Derive the recursive equation $\termholtcnat{f}=\whiteboxtransport{(\lambda\vec{x}.\app \termholtc{t})}$ using that $\relnat$ is left-unique.
\end{enumerate*}
The use of black-box transports relies on the fact that $\relnat$ can be used as a relation by \transport~\cite{transport}:
\begin{theorem}[Blackbox-Transport~\cite{transport}]\label{thm:blackbox}
$\relnat$, $\natify$, and $\denatify$ form a partial Galois equivalence.
Thus, there is some $\termholtcnat{f}$ that is $\relnat$-related to $\termholtc{f}$.
\end{theorem}
In the formalisation, we use the \transport prototype \cite{transport-AFP} to obtain $\termholtcnat{f}$.
Since the black-box transport disregards the structure of $t$,
\cref{thm:blackbox} may not preserve tail-recursiveness and first-order applications.
To derive an equation $\termholtcnat{f}=\lambda\vec{\termholtcnat{x}}.\app\termholtcnat{t}$
structurally related to the one of $\termholtc{f}$,
we use \makebox{white-box~transport}:
\begin{definition}\label{def:typerelnatsubst}
The \emph{white-box-transport} of $\lambda\vec{x}.\app \termholtc{t}$
replaces every subterm by its $\relnat$-related companion.
It is defined recursively by
\begin{gather*}
\whiteboxtransport{(\lambda \vec{x}.\app s)} \define \lambda \vec{\termholtcnat{x}}.\app \whiteboxtransport{s},\quad\hspace{0.5em}
\whiteboxtransport{(t_1\app t_2)} \define \whiteboxtransport{t_1}\app\whiteboxtransport{t_2},\quad\hspace{0.5em}
\whiteboxtransport{x_i} \define \termholtcnat{x_i},\quad\hspace{0.5em}
\whiteboxtransport{g} \define \termholtcnat{g},\quad\hspace{0.5em}
\whiteboxtransport{f} \define \termholtcnat{f}.
\end{gather*}
\end{definition}
\begin{theorem}[Whitebox-Transport~\cite{huffmanTransfer}]\label{thm:whitebox}
$\whiteboxtransport{(\lambda\vec{x}.\app \termholtc{t})}$ and $\lambda\vec{x}.\app \termholtc{t}$ are $\relnat$-related
and white-box transports preserve tail-recursiveness and first-order applications.
\end{theorem}
\begin{theorem}\label{thm:cond_eq_natify_funs}
If $\vec{\relnat\app \termholtcnat{x_i}\app x_i}$ then $\termholtcnat{f}\app\vec{\termholtcnat{x}}=\whiteboxtransport{t}$.
\end{theorem}
\begin{proof}
We have
$\relnat\app(\termholtcnat{f}\app\vec{\termholnat{x}})\app(\termholtc{f}\app\vec{\termholtc{x}})$
and
$\relnat\app(\whiteboxtransport{t})\app(\termholtc{f}\app\vec{\termholtc{x}})$
by \cref{thm:blackbox,thm:whitebox}.
Moreover, $\relnat$ is left-unique by~\cref{def:nat_encoded}.
Thus $\termholtcnat{f}\app\vec{\termholnat{x}}=\whiteboxtransport{t}$.
\end{proof}
\cref{thm:blackbox,thm:whitebox,thm:cond_eq_natify_funs} show that we have achieved~\cref{obj:holtctoholtcnat}.
The described synthesis is fully automatic in our framework. An example is shown in \cref{fig:holtctoholtcnat_steps}.
\begin{figure}[t]
\begin{subfigure}[t]{0.5\linewidth}
\lstset{caption={Input function from the user.},label={lst:holtctoholtcnat_steps1}}
\begin{lstlisting}[language=isabelle]
count [*a*] [*xs*] [*n*] = case [*xs*] of [] $\Rightarrow$ [*n*]
| [**x**] # [**xs**] $\Rightarrow$ count [*a*] [**xs**]
    (if [**x**] = [*a*] then Suc [*n*] else [*n*])
\end{lstlisting}
\end{subfigure}
\begin{subfigure}[t]{0.5\linewidth}
\lstset{caption={Constant {\ttfamily count}$^\natty$ and relatedness theorem from black-box transport.},
label={lst:holtctoholtcnat_steps2}
}
\begin{lstlisting}[language=isabelle]
($\funrel{\relnat}{\funrel{\relnat}{\funrel{\relnat}{{\relnat}}}}$) count$^\natty$ count
\end{lstlisting}
\end{subfigure}
\begin{subfigure}[t]{0.5\linewidth}
\lstset{caption=White-box transport theorem.,label={lst:holtctoholtcnat_steps3}}
\begin{lstlisting}[language=isabelle]
($\funrel{\relnat}{\funrel{\relnat}{\funrel{\relnat}{{\relnat}}}}$)
($\lambda$[**a$^{\color{isabvar}\natty}$**] [**xs$^{\color{isabvar}\natty}$**] [**n$^{\color{isabvar}\natty}$**].  case$^\natty$ [**xs$^{\color{isabvar}\natty}$**] of []$^\natty$ $\Rightarrow$ [**n$^{\color{isabvar}\natty}$**]
| [**x$^{\color{isabvar}\natty}$**] #$^\natty$ [**xs$^{\color{isabvar}\natty}$**] $\Rightarrow$ count$^\natty$ [**a$^{\color{isabvar}\natty}$**] [**xs$^{\color{isabvar}\natty}$**]
  (if [**x$^{\color{isabvar}\natty}$**] = [**a$^{\color{isabvar}\natty}$**] then Suc$^\natty$ [**n$^{\color{isabvar}\natty}$**] else [**n$^{\color{isabvar}\natty}$**]))
count
\end{lstlisting}
\end{subfigure}
\begin{subfigure}[t]{0.5\linewidth}
\lstset{caption={Final tail-recursive equation}.,label={lst:holtctoholtcnat_steps4}}
\begin{lstlisting}[language=isabelle]
$\relnat$ [*a$^{\color{isafvar}\natty}$*] [*a*] $\land$ $\relnat$ [*xs$^{\color{isafvar}\natty}$*] [*xs*] $\land$ $\relnat$ [*n$^{\color{isafvar}\natty}$*] [*n*]$\implies$
count$^\natty$ [*a$^{\color{isafvar}\natty}$*] [*xs$^{\color{isafvar}\natty}$*] [*n$^{\color{isafvar}\natty}$*] = case$^\natty$ [*xs$^{\color{isafvar}\natty}$*] of
  []$^\natty$ $\Rightarrow$ [*n$^{\color{isafvar}\natty}$*]
| [**x$^{\color{isabvar}\natty}$**] #$^\natty$ [**xs$^{\color{isabvar}\natty}$**] $\Rightarrow$ count$^\natty$ [*a$^{\color{isafvar}\natty}$*] [**xs$^{\color{isabvar}\natty}$**]
  (if [**x$^{\color{isabvar}\natty}$**] = [*a$^{\color{isafvar}\natty}$*] then Suc$^\natty$ [*n$^{\color{isafvar}\natty}$*] else [*n$^{\color{isafvar}\natty}$*])
\end{lstlisting}
\end{subfigure}
\caption{%
For the user-specified function\protect\footnotemark\ \lstinline{count} in \cref{lst:holtctoholtcnat_steps1},
we obtain a related constant {\ttfamily count}$^\natty$ in \cref{lst:holtctoholtcnat_steps2}
using~\cref{thm:blackbox}.
Using \cref{thm:whitebox}, we obtain a second term related to {\ttfamily count} in \cref{lst:holtctoholtcnat_steps3}.
Using that $\relnat$ is left-unique (cf.\ \cref{thm:cond_eq_natify_funs}),
we derive the desired equation for
{\ttfamily count}$^\natty$ in \cref{lst:holtctoholtcnat_steps4}.
}\label{fig:holtctoholtcnat_steps}
\end{figure}
\footnotetext{%
There is a commmand in \isabellehol that turns functions specified by multiple equations
into an equivalent definition consisting of only one equation.}

\renewcommand{\holtc}{\oldholtc}
\renewcommand{\holtcnat}{\oldholtcnat}
\renewcommand{\termholtc}{\oldtermholtc}
\renewcommand{\termholtcnat}{\oldtermholtcnat}

\section{$\foltcnat$ to $\imptc$}\label{sec:holtcnattoimptc}

We next describe the deep embedding from $\foltcnat$ to $\imptc$.
\cref{fig:imptc_rules} shows the $\imptc$-specific commands and their semantics,
extending those in \cref{fig:exec_imp}.
Notably, $\imptc$ supports calls of $\impw$ programs (defined in \cref{sec:impctoimpw}) and tail-recursion.
It also supports if-then-elses but no case combinators.
Since every $\casetypenat{d}$ is defined as a nested if-then-else construction
(\cref{def:natify_constructor}),
this is sufficient.
\begin{figure}[t]
\begin{mathpar}
  \inferrule*[left=\callN]
  {\bigstepTimePart{pc}{s}{n}{r}{v} \\
   s'=\funupd{s}{r}{v}}
  {\bigstepCtxtTime{p}{\call{pc}{r}}{s}{n}{s'}}[\callN]\label{infrule:call}

  \inferrule*[left=\tailN]
  {\bigstepCtxtTime{p}{p}{s}{n}{s'}}
  {\bigstepCtxtTime{p}{\tail}{s}{n+\const}{s'}}[\tailN]\label{infrule:tail}
\end{mathpar}
\caption{Execution relation of $\imptc$-specific commands.}\label{fig:imptc_rules}
\end{figure}

To relate a $\foltcnat$ function $\termfoltcnat{f}$, with its conditional (recursive) equation
$\vec{\relnat\app \termfolnat{x_i}\app x_i} \implies \termfoltcnat{f}\app \vec{\termfolnat{x}} = \termfoltcnat{t}$,
to an $\imptc$ program $p$,
we presume injective functions $\htoiargconst,\htoiretconst$ such that
\begin{enumerate*}
\item $\htoiargi{f}{i}$ returns a unique name for the $i$-th argument of $f$ and
\item $\htoiret{f}$ returns a unique name for the result computed by~$f$.
\end{enumerate*}

We focus on total, functional correctness of~$p$ in this work
(proving concrete time bounds for $p$ is left as future work, see~\cref{sec:concl}):
From any state $s$ with well-encoded inputs $\vec{\stateretr{s}{\htoiargi{f}{i}}}$,
the program $p$ terminates and the value in return register $\htoiret{f}$
is equal to $\termfoltcnat{f}\app \vec{\stateretr{s}{\htoiargi{f}{i}}}$.
Formally, our objective is:
\begin{objective}\label{obj:holtcnattoimptc}
Given a $\foltcnat$ function $\termfoltcnat{f}$,
compile it to an $\imptc$ program $p$ such that
$\bigstepCtxtPart{p}{p}{s}{\htoiret{f}}{\termfoltcnat{f}\app \vec{\stateretr{s}{\htoiargi{f}{i}}}}$
whenever $\vec{\relnat\app (\stateretr{s}{\htoiargi{f}{i}})\app x_i}$.
\end{objective}
The deep embedding has two steps:
\begin{enumerate*}
\item Compile $\termfoltcnat{f}$ into an $\imptc$ program $p$ (\cref{sec:holtcnattoimptc_compiler}) and
\item prove the equivalence between $\termfoltcnat{f}$ and $p$ using custom-built automation (\cref{sec:holtcnattoimptc_correctness}).
The automation uses symbolic execution,
normalises the (otherwise incomprehensible) program state such that it becomes amenable to automatic proof,
and discharges well-encodedness side-conditions.
\end{enumerate*}

\paragraph{Note on Higher-Order Functions}
Since $\imptc$ does not support higher-order functions,
we assume first-order functions as input for our framework.
Nevertheless, each first-order instance of a higher-order function can be compiled systematically with our framework.
We demonstrate this by means of an example.
More cases can be found in the formalisation:
\begin{example}\label{ex:compile_higher_order}
Consider the iteration function
of type \makebox{$(\alphaty\purefun\alphaty)\purefun\natty\purefun\alphaty\purefun\alphaty$}:
\begin{gather}\label{eq:fun_pow}
f^n\app x\define \caseconst\app n\app\ofconst\app 0 \Rightarrow x\mid n+1\Rightarrow f^n\app(f\app x).
\end{gather}
This function cannot be compiled as it is.
However, every first-order instance,
e.g.\ $\sndnat^j$ as used in the definition of $\natselector{i}{j}$
(\cref{def:natify_constructor}),
can be compiled.
To do so, we define $\powsndnat\app n\define \sndnat^n$.
We then instantiate $\cref{eq:fun_pow}$ with $f=\app \sndnat$
and fold the definition of $\powsndnat$ to obtain
$\powsndnat\app n\app x=\caseconst\app n\app\ofconst\app 0 \Rightarrow x\mid \allowbreak n+1\Rightarrow \powsndnat\app n\app(\sndnat\app x)$.
This equation can then be compiled with our framework.
\end{example}

\subsection{Compilation to $\imptc$}\label{sec:holtcnattoimptc_compiler}

We compile $\termfoltcnat{f}\app \vec{\termfolnat{x}} = \termfoltcnat{t}$
to an $\imptc$ program $p$ in two steps.
For ease of notation, we drop the superscript and just write $f$ and $t$ below.

First, we
parse $t$ into a metaprogram datatype resembling $\foltcnat$.\footnote{%
While the datatype allows recursive calls in non-tail positions, the compiler rejects such terms using a syntactic check.}
Second, we compile the parsed term to an $\imptc$ program, as shown in \cref{fig:holtcnat_compiler}.
\begin{figure}[t]
\begin{gather*}
\begin{aligned}
\dens{\holnIf{t_1}{t_2}{t_3}}{b}{r} \define &\; \dens{t_1}{b}{x} \seqi \iif{x}{\dens{t_2}{b}{r}}{\dens{t_3}{b}{r}} & \text{(fresh $x$)} \\
\dens{\holnLet{t_1}{t_2}}{b}{r} \define &\; \dens{t_1}{b}{x} \seqi \dens{t_2}{x \# b}{r} & \text{(fresh $x$)} \\
\dens{\holnLetBound{n}}{b}{r} \define &\; \assign{r}{b\app {!}\app n} & \\
\dens{\holnArg{n}}{b}{r} \define &\; \assign{r}{\htoiargi{f}{n}} & \\
\dens{\holnNum{n}}{b}{r} \define &\; \assign{r}{n} & \\
\dens{\holnCall{g}{t_1, \dotsc, t_m}}{b}{r} \define &\; \dens{t_1}{b}{x_1} \seqi {\dotsc} \seqi \dens{t_m}{b}{x_m} \seqi & \text{(fresh $x_1, \dotsc, x_m$)}  \\
    &\; \assign{\htoiargi{g}{1}}{x_1} \seqi \dotsc \seqi \assign{\htoiargi{g}{m}}{x_m} \seqi & \\
    &\; \call{g^{\imp}}{\htoiret{g}} \seqi  \assign{r}{\htoiret{g}} &\text{($g^{\imp}$ registered for $g$)}\\
\dens{\holnTailCall{t_1, \dotsc, t_k}}{b}{r} \define &\; \dens{t_1}{b}{x_1} \seqi {\dotsc} \seqi \dens{t_k}{b}{x_k} & \text{(fresh $x_1, \dotsc, x_k$)} \\
    &\; \assign{\htoiargi{f}{1}}{x_1} \seqi \dotsc \seqi \assign{\htoiargi{f}{k}}{x_k} \seqi \tail\hspace{-1em} &
\end{aligned}
\end{gather*}
\caption{The compiler from the $\foltcnat$ representation to $\imptc$.
Term $t_1$ in $\holnIf{t_1}{t_2}{t_3}$ denotes the condition $t_1\neq 0$,
$\holnLet{t_1}{t_2}$ binds $t_1$ to the first de Bruijn index in $t_2$,
$\holnLetBound{i}$ denotes the variable bound by the $i$-th enclosing \textbf{Let}, and
$\holnArg{i}$ the $i$-th argument of $f$, i.e.\ $x_i$.
}\label{fig:holtcnat_compiler}
\end{figure}
We denote the compilation by $\dens{t}{b}{r}$,
where $b$ is a list of registers holding values of terms bound by enclosing $\holnLetC{}$ bindings and
$r$ is the register that will hold the program's result.

For a call $\holnCall{g}{t_1, \dotsc , t_m}$,
the compiler retrieves the registered $\impw$ implementation $g^{\imp}$ of $g$.
For the primitives equality, addition, and subtraction,
manual implementations must be provided to the compiler.
The compiler generates fresh register names when needed, namely
to compile $\holnLetC{}$ bindings,
to store the value of $\holnIfC{}$ conditions,
and to save function arguments in temporary registers to prevent them from being overwritten.
For example, in $\holnCall{g}{t_1, \dotsc , t_m}$,
argument registers could be overwritten if another call of $g$ appears within $t_1, \dotsc , t_m$.

In summary, we parse $t$ into a term $t^\prime$ of the metaprogram datatype
and compile $t'$ to the $\imptc$ program $p\define \dens{t^\prime}{[]}{\htoiret{f}}$.

\subsection{Correctness Proofs}\label{sec:holtcnattoimptc_correctness}
Since the compiler to $\imptc$ is a metaprogram,
its correctness cannot be proven in \isabellehol.
Instead, the correctness of each compiled program $p$ must be verified individually through automation
(a process sometimes called ``certified extraction''), which we describe next.

First, $p$ is normalised such that no recursive constructor (sequences and if-then-elses) appears on the left of a sequence.
This simplifies the implementation, but shall not concern us any further
(for details see~\appendixref{sec:holtcnattoimptc_correctness_details}).
The automation proceed as follows:
\begin{enumerate}
  \item If $f$ is recursive, start an induction proof on $\vec{x}$ using the induction rule of $f$.
\item Symbolically run $p$ to completion or until a recursive call is encountered.
Let $s'$ be the end state.
\item Run $\termfoltcnat{f}$ to completion with result $r$ or until a recursive call $\termfoltcnat{f}\app \vec{y}$ is encountered.
\item If both computations ran to completion, prove $\stateretr{s'}{\htoiret{f}}=r$.
Otherwise, prove $\vec{\stateretr{s'}{\htoiargi{f}{i}}}=\vec{y}$ and conclude using the inductive hypothesis.
\end{enumerate}
\begin{figure}[t]
\begin{mathpar}
\inferrule*[left=Update]
  {\statenorm{\aval{a}{s}}{v}}
  {\statenorm{\funupd{s}{r}{\aval{a}{s}}}{\funupd{s}{r}{v}}}

\inferrule*[left=Const]
  {$ $}
  {\statenorm{\aval{\aconst{n}}{s}}{n}}

\inferrule*[left=Reg]
  {\statenorm{\stateretr{s}{r}}{v}}
  {\statenorm{\aval{\avar{r}}{s}}{v}}

\inferrule*[left=Hit]
  {r=r'}
  {\statenorm{\stateretr{\funupd{s}{r}{v}}{r'}}{v}}\hspace{0.8em}
\inferrule*[left=Cont]
  {r\neq r'\\\statenorm{\stateretr{s}{r'}}{v'}}
  {\statenorm{\stateretr{\funupd{s}{r}{v}}{r'}}{v'}}\hspace{0.8em}
\inferrule*[left=Arith]
  {\statenorm{\aval{A_1}{s}}{v_1}\\
  \statenorm{\aval{A_2}{s}}{v_2}}
  {\statenorm{\aval{A_1\otimes A_2}{s}}{v_1\otimes v_2}}
\end{mathpar}
\caption{State normalisation rules.}\label{fig:statenorm}
\end{figure}
To run $p$, we use the execution rules from \cref{fig:exec_imp,fig:imptc_rules}
(specialised to single return registers)\footnote{The specialisation is routine and can be found in \cref{sec:holtcnattoimptc_correctness_details}.}
and the correctness theorems of subprograms.
Each update to the $\imptc$ state is normalised with the rules from \cref{fig:statenorm} to
\begin{enumerate*}
\item to simplify the proof state for interactive proofs and
\item to speed up the automation.
\end{enumerate*}
The automation eliminates cases where $p$ and $\termfoltcnat{f}$ took different execution branches.

Since our correctness theorems are conditional,
the inductive hypothesis and the subprograms' correctness theorems are conditional too.
Whenever they are used, goals of the form
$\relnat\app \termholnat{x}\app x$ must be solved.
This is done automatically by a resolution-based method that
uses $\cref{def:relation}$,
the relatedness theorems of ADTs and functions (\cref{thm:nat_encoded_datatype,thm:blackbox}),
and local assumptions.
Note that this is essentially a form of type-checking,
making sure that all involved terms are well-encoded.
\begin{example}\label{ex:tc_to_imptc}
We continue with the example from \cref{fig:holtctoholtcnat_steps}
and show the correctness of {\ttfamily count$^{\natty}$}.
The goal states are shown in \cref{fig:holtcnattoimptc_correctness_steps}.
\begin{figure}[t]
\begin{subfigure}[t]{\linewidth}
\lstset{caption={Initial goal.}}
\begin{lstlisting}[language=isabelle]
$\relnat$ ([*s*]  "a") [*a*] $\land$ $\relnat$ ([*s*] "xs") [*xs*] $\land$ $\relnat$ ([*s*] "n") [*n*] $\implies$
$\vdash$ (count$^{\imp}$, [*s*]) $\purefun_{ret}$ (count$^{\natty}$ ([*s*] "a") ([*s*] "xs") ([*s*] "n"))
\end{lstlisting}
\end{subfigure}
\begin{subfigure}[t]{\linewidth}
\lstset{caption={\ttfamily (\#)}-case after induction setup.}
\begin{lstlisting}[language=isabelle]
$\relnat$ ([*s*]  "a") [*a*] $\land$ $\relnat$ ([*s*] "xs") ([*x*] # [*xs*]) $\land$ $\relnat$ ([*s*] "n") [*n*] $\implies$ (*$\color{isabtextcomment}\relnat$-assms.*)
($\forall{}$[**s'**]. $\relnat$ ([**s'**] "a") [*a*] $\land$ $\relnat$ ([**s'**] "xs") [*xs*] $\qquad\qquad$ (*inductive hypothesis*)
  $\land$ $\relnat$ ([**s'**] "n") (if [*x*] = [*a*] then Suc [*n*] else [*n*]) $\longrightarrow$
  $\vdash$ (count$^{\imp}$, [**s'**]) $\purefun_{ret}$ (count${}^{\natty}$ ([**s'**] "a") ([**s'**] "xs") ([**s'**] "n"))) $\implies$
$\vdash$ (count$^{\imp}$, [*s*]) $\purefun_{ret}$ (count$^{\natty}$ ([*s*] "a") ([*s*] "xs") ([*s*] "n"))
\end{lstlisting}
\end{subfigure}
\begin{subfigure}[t]{\linewidth}
\lstset{caption=Subgoal after symbolic execution (repeated premises abbreviated as {\ttfamily <...>}).}
\begin{lstlisting}[language=isabelle]
<$\relnat$-assumptions> $\land$ <inductive hypothesis> $\land$ [*x*] = [*a*] $\implies$
$\vdash$ (count$^{\imp}$, [*s'*]) $\purefun_{ret}$ (count$^{\natty}$ ([*s*] "a") (select$_{2,2}$ ([*s*] "xs")) (Suc ([*s*] "n")))
\end{lstlisting}
\end{subfigure}
\begin{subfigure}[t]{\linewidth}
\lstset{caption={Goals after application of inductive hypothesis}.}
\begin{lstlisting}[language=isabelle]
$\relnat$ ([*s*] "a") [*a*] $\land$ $\relnat$ (select$_{2,2}$ ([*s*] "xs")) [*xs*] $\land$ $\relnat$ (Suc ([*s*] "n")) (Suc [*n*])
\end{lstlisting}
\end{subfigure}
\caption{Step-by-step correctness proof for {\ttfamily count$^{\natty}$}, described in \cref{ex:tc_to_imptc}.}\label{fig:holtcnattoimptc_correctness_steps}
\end{figure}
First, the induction rule of \lstinline{count} is applied,
creating cases for \lstinline{[]} and \lstinline{(#)}.
We focus on the latter.
Running \lstinline{count$^{\imp}$} and \lstinline{count$^{\natty}$}
produces two subcases: one for \lstinline{[*x*] $=$ [*a*]} and one for \lstinline{[*x*] $\neq$ [*a*]}.
We consider the former.
The new state \lstinline{[*s'*] = [*s*](r1 := x1,...)}
is obtained from \lstinline{[*s*]} by normalised state updates.
Due to normalisation, we immediately get:
\begin{lstlisting}[language=isabelle]
[*s'*] "a" = [*s*] "a" $\land$ [*s'*] "xs" = select$_{2,2}$([*s*] "xs") $\land$ [*s'*] "n" = Suc([*s*] "n")
\end{lstlisting}
Thus, the inductive hypothesis can be applied, leaving a set of relatedness goals.
These are solved automatically,
using the local assumptions,
the selector relatedness theorems for lists,
and the relatedness theorem for \lstinline{Suc}.
\end{example}

\paragraph{Case Studies}
We applied our framework to typical functional programs,
including standard datatype functions ($\map$, $\fold$,\dotso),
the bisection method for square roots,
and problems from complexity theory --
such as reductions from \sat to \threesat and \threesat to independent sets.
The examples can be found in the formalisation.
All proofs were automatic, but functions using hundreds of registers
had to be split due to inefficient retrievals for the tracked $\imptc$ state in the current implementation.
Performance improvements are future work.
We conjecture that our method is complete when applied to programs compiled from $\foltc$ using our framework,
but a formal proof is out of scope.

\section{$\imptc$ to $\impminus$}

In this section, we describe the verified chain of compilers from $\imptc$ to $\impminus$:
\begin{objective}\label{obj:imptctoimpminus}
For every $\imptc$ program $p$,
compile an $\impminus$ program $p'$ such that (up to encoding and decoding), for the same input,
\begin{enumerate*}
\item $p$ and $p'$ compute the same output and
\item the number of steps $p'$ takes is at most polynomially more than the number of steps $p$ takes.
\end{enumerate*}


\end{objective}
Note that since the size of the values stored in registers is governed by \cref{thm:spaceConsume}, if the above objective is satisfied,
the memory used by the final $\impminus$ program is polynomial in the size of $p$ and its running time.

The compilation proceeds in several steps:
First, we compile recursive calls into while-loops of an intermediate language $\impc$ (\cref{sec:imptctoimpc}),
Second, we eliminate calls to other programs (\cref{sec:impctoimpw}).
Finally, we use bit blasting to compile into our final language $\impminus$ (\cref{sec:impwtoimpminus}).

\subsection{$\imptc$ to $\impc$}\label{sec:imptctoimpc}
In the first step, we replace tail-recursive calls by while-loops,
whose semantics are given in \cref{fig:whilerules}.
\begin{figure}[t]
  \begin{mathpar}
    \inferrule*[left=\whileFalseN]
    {\stateretr{s}{r} = 0}
    {\bigstepTime{\while{r}{p}}{s}{\const}{s}}
    [\whileFalseN]\label{infrule:whileFalse}\hspace{0.3em}
  \inferrule*[left=\whileTrueN]
    {\stateretr{s}{r} \neq 0\quad
     \bigstepTime{p}{s}{n_1}{s'}\quad 
     \bigstepTime{\while{r}{p}}{s'}{n_2}{s''}}
    {\bigstepTime{\while{r}{p}}{s}{n_1+n_2+\const}{s''}}
    [\whileTrueN]\label{inferrule:whileTrue}
  \end{mathpar}
  \caption{Execution relation for $\whileop$.}\label{fig:whilerules}
\end{figure}
In rule \nameref{inferrule:whileTrue},
the loop-condition $r\ \impconstrfont{\neq 0}$ is checked
\emph{before} stepping into the $\whileop$.
In contrast, tail-recursive calls occur at the \emph{end} of the program.
To bridge this gap, we prove the following rules for $\imptc$.
\begin{lemma}\label{lem:tailWhileRules}
Let $\bigstepTime{p}{s}{n}{(s',b)}$
denote the modified execution relation shown in \cref{fig:exectailcallflag} that runs
$p$ but stops and flags whenever a tail-recursive call is encountered using \makebox{$b\in\{\false,\true\}$}
\begin{figure}[t]
\begin{mathpar}
  \inferrule*
  {\stateretr{s}{r} \neq 0 \\ \bigstepTime{p_1}{s}{n}{(s',b)}}
  {\bigstepTime{\iif{r}{p_1}{p_2}}{s}{n+\const}{(s',b)}}

  \inferrule*
  {\bigstepTime{p_1}{s}{n_1}{(s', \false)}\\
  \bigstepTime{p_2}{s'}{n_2}{(s'', b)}}
  {\bigstepTime{\seq{p_1}{p_2}}{s}{n_1+n_2+\const}{(s'',b)}}

  \inferrule*
  {$ $}
  {\bigstepTime{\tail}{s}{\const}{(s,\true)}}

  \inferrule*
  {s'=\funupd{s}{r}{\aval{a}{s}}}
  {\bigstepTime{\assign{r}{a}}{s}{\const}{(s',\false)}}

  \raisebox{1em}{\dots}
\end{mathpar}
\caption{Execution relation from \cref{lem:tailWhileRules} that halts with a flag when encountering a tail-recursive call. Omitted rules are routine and listed in
\appendixref{sec:imptctoimpc_details}.}\label{fig:exectailcallflag}
\end{figure}
Then the following rules are admissible:
\begin{mathpar}
\inferrule*
  {\bigstepTime{p}{s}{n}{(s',\false)}}
  {\bigstepCtxtTime{tp}{p}{s}{n}{s'}}

\inferrule*
  {\bigstepTime{p}{s_1}{n_1}{(s_2,\true)}\\
   \bigstepCtxtTime{tp}{tp}{s_2}{n_2}{s_3}}
  {\bigstepCtxtTime{tp}{p}{s_1}{n_1+n_2}{s_3}}
\end{mathpar}
\end{lemma}
Note the similarity between these rules and those for \nameref{inferrule:whileTrue} and \nameref{infrule:whileFalse}.
These rules suggest a straightforward translation from $p$ to an $\impc$ program
by using a loop that continues as long as a recursive call would have occurred:
\begin{theorem}\label{thm:comp_t_correct}
Let $\impconstrfont{cnt}$ be a fresh register with respect to $\vars\app p$.
Let $p\psubst{t}{x}$ denote the syntactic substitution of any occurrence of $x$ by $t$ in a program $p$.
Define
\begin{gather*}
\imptctoimpc{p} \define
\seq{\assign{\impconstrfont{cnt}}{\impconstrfont{1}}}{\while{\impconstrfont{cnt}}{\seq{\assign{\impconstrfont{cnt}}{\impconstrfont{0}}}{p\psubst{(\assign{\impconstrfont{cnt}}{\impconstrfont{1}})}{\tail}}}}.
\end{gather*}
Then it holds that
$\bigstepCtxtTimePart{p}{p}{s}{n}{\vars{p}}{s'}$ \textiff $\bigstepTimePart{\imptctoimpc{p}}{s}{n+\const}{\vars{p}}{s'}$.
\end{theorem}
Since we use an additional register, we only get partial equivalence of states in the correctness theorem.
This is sufficient for our purposes nonetheless.

\subsection{$\impc$ to $\impw$}\label{sec:impctoimpw}
We next eliminate calls to $\impw$ programs.
We simply inline every call to a program $p^w$
by copying its used registers to fresh memory locations,
executing the program there, and copying the result register back:
\begin{definition}
Let $\vars\app p^w=\{r_1,\dotsc,r_n\}$ and $m$ be a renaming to fresh registers.
We define the inlining of a call to program $p^w$ with return register $r$ as
\begin{gather*}
\impctoimpw{\call{p^w}{r}} \define \seq{\seq{\assign{m\app r_1}{r_1}}{\seq{\dotsb}{\seq{\assign{m\app r_n}{r_n}}}{p^w\psubst{m\app x}{x}}}}{\assign{r}{m\app r}}.
\end{gather*}
For an $\impc$ program $p$, we let $\impctoimpw{p}$ denote the $\impw$ program obtained from $p$ by inlining all calls.
\end{definition}
Note that called programs must already be compiled to $\impw$ and hence do not contain further calls.
Arguing about a program in a different memory location requires the obvious substitution lemma:
\begin{lemma}\label{lem:neat_subst}
$\bigstepTime{p\psubst{(m\app x)}{x}}{s}{n}{s'}$ implies $\bigstepTime{p}{s\comp m}{n}{s'\comp m}$
for injective $m$.
\end{lemma}
The inlining induces a blow-up that is dependent on the program size (but not the input), as shown below.
\begin{theorem}\label{thm:comp_c_correct}
$\bigstepTimePart{p}{s}{n}{\vars\app p}{s'}$ if $\bigstepTimePart{\impctoimpw{p}}{s}{f(n,p)}{\vars\app p}{s'}$
where $f(n,p)$ is linear in $n \cdot \abs{p}$.
Moreover, if $\bigstepTimePart{\impctoimpw{p}}{s}{n}{\vars\app p}{s'}$ then $\bigstepTimePart{p}{s}{n}{\vars\app p}{s'}$.
\end{theorem}

Note that the fact that $\impw$ is reasonable w.r.t.\ time and space follows from standard reductions in computability theory.
A good exposition is in Erk and Priese's book~\cite{erkTheo}.
There it is shown that While-programs can be simulated by Goto-programs with constant space (one-to-one-correspondence between registers) and polynomial time overheads (Theorem 12.4.3), and Goto-programs can be simulated with constant space (one tape per register) and polynomial time overheads by $k$-tape Turing Machines (Theorem 11.4.1), and $k$-tape TMs can be simulated by single-tape TMs in constant space ($2k$ tuple letters per cell) and quadratic running time overheads (Theorem 7.4.5).
The other direction, showing that TMs can be simulated by While-programs is obvious.

\subsection{$\impw$ to $\impminus$}\label{sec:impwtoimpminus}
\begin{figure}[t]
  \begin{mathpar}
    \inferrule*
    {b\in\{\false,\true\}\\
    s'=\funupd{s}{r}{b}}
    {\bigstepTime{\assign{r}{b}}{s}{1}{s'}}

    \inferrule*
    {\stateretr{s}{r} \neq \false \\ \bigstepTime{p_1}{s}{n}{s'}}
    {\bigstepTime{\iif{r}{p_1}{p_2}}{s}{n+1}{s'}}

    \inferrule*
    {\bigstepTime{p_1}{s}{n_1}{s'} \\
    \bigstepTime{p_2}{s'}{n_2}{s''}}
    {\bigstepTime{\seq{p_1}{p_2}}{s}{n_1+n_2+1}{s''}}

    \inferrule*
    {\stateretr{s}{r} = \false \\ \bigstepTime{p_2}{s}{n}{s'}}
    {\bigstepTime{\iif{r}{p_1}{p_2}}{s}{n+1}{s'}}

    \inferrule*
    {\stateretr{s}{r} = \false}
    {\bigstepTime{\while{r}{p}}{s}{1}{s}}\hspace{1em}
    \inferrule*
    {\stateretr{s_1}{r} \neq \false\\
    \bigstepTime{p}{s_1}{n_1}{s_2}\\
    \bigstepTime{\while{r}{p}}{s_2}{n_2}{s_3}}
    {\bigstepTime{\while{r}{p}}{s_1}{n_1+n_2+2}{s_3}}
  \end{mathpar}
  \caption{Execution relation for $\impminus$.}\label{fig:imp_minus_def}
\end{figure}
Our programs so far operate on arbitrary-size natural numbers.
In contrast, our final language $\impminus$ operates on single bits,
represented as $\false$ and $\true$.
The commands of $\impminus$ and their semantics are shown in \cref{fig:imp_minus_def}.
There are no arithmetic expressions, and all registers represent a single bit.

We use bit-blasting to compile $\impw$ programs
to $\impminus$ programs.
The bit-blasting is parametrised with a number $w$, indicating the bit-width of the compilation.
Each $\impw$ register $r$ is represented as $w+1$ fresh $\impminus$ registers $r_0,\dotsc,r_w$:
The $\impminus$ register $r_{i}$ corresponds to the $i$-th bit of the $\impw$ register $r$,
and the register $r_0$ indicates whether all bits are $\false$.
Arithmetic expressions of $\impw$ are bit-blasted to $\impminus$ programs.
We showcase the bit-blasting for addition here.
The addition of $x$ and $y$ in register $r$ is compiled to
\begin{gather}
\impminusminusadd{x}{y}{r}{w} \equiv \seq{\copyopnd{a}{x}{w}}{\seq{\copyopnd{b}{y}{w}}{\adder{r}{w}}},
\end{gather}
where $a,b$ are fresh, fixed $\impw$ registers.
The function copies the first $w$ bits of $x$ and $y$ to $a_1,\dotsc,a_w$ and $b_1,\dotsc,b_w$, respectively,
which are then added by the function $\adder{r}{w}$.
The function $\adder{r}{w}$ puts the result of the addition in registers $r_0,\dotsc,r_w$.
It is implemented as follows:
\begin{gather}
\adder{r}{w} \define \fulladder{r}{1} ; \dotsb;  \fulladder{r}{w} ; \assign{carry}{\false} ; \assign{zero}{\false},
\end{gather}
where $\fulladder{r}{i}$ is a program that implements a standard full adder with zero check.
It adds the bits $a_i$ and $b_i$ and the carry register $carry$, storing the result in $r_i$, the new carry value in $carry$, and stores in $zero$
whether the result and the previous $zero$ value were $\false$.
The full adder is called for all bits of the addends.
Then $r_0$ is written, and the $carry$ and $zero$ registers are reset.

We bit-blast all other arithmetic expressions of $\impw$ in a similar way;
details are in the formalisation.
In the next definition, we fix a maximum bit-width $w$:
\begin{definition}\label{def:bitblast}
We denote with $\bitblastAexp{a}{w}{r}$ the bit-blasted version of arithmetic expression~$a$ stored in register $r$.
For an $\impw$ state $s$, we denote with $\bitblastState{s}{w}$ its corresponding $\impminus$ state,
i.e.\ $\stateretr{\bitblastState{s}{w}}{r_i} = (\stateretr{s}{r})_i$ for all registers $r$ and $i < w$,
where $n_i$ denotes the $i$-th bit of $n$.
\end{definition}
The functional correctness can now be stated as follows:
\begin{lemma}\label{lem:bitblastExpr}
If $\maxset\{\maxconst{a}, \aval{a}{s}, \maxconst{s}\}< 2^w$
then $\bigstepTime{\bitblastAexp{a}{w}{r}}{\bitblastState{s}{w}}{f(w)}{\bitblastState{\funupd{s}{r}{\aval{a}{s}}}{w}}$,
where $f(w)$ is linear in $w$.
\end{lemma}
Lastly, we show that the overall runtime blow-up of the bit-blasted program is polynomial.
We first lift the bit-blasting of expressions to $\impw$ programs $p$, denoted by $\bitblast{p}{w}$.
In $\bitblast{p}{w}$, arithmetic expressions are replaced with their bit-blasted version
and all other constructs are mapped to their corresponding constructs in the obvious way.
For functional correctness, we get:
\begin{theorem}\label{lem:bitblastProgram}
If $\bigstepTime{p}{s}{n}{s'}$
and $n<w$
and $\maxset\{\maxconst{s},\maxconst{p}\} \cdot 2^n < 2^w$,
then $\bigstepTime{\bitblast{p}{w}}{\bitblast{s}{w}}{f(n,w)}{\bitblast{s'}{w}}$ where $f(n,w)$ is linear in $n\cdot w$.
\end{theorem}
\begin{proof}[Proof sketch]
The proof is by structural induction on the $\impminus$ program $p$.
There are two interesting cases.
The first is the case of assignment, where the theorem follows from Lemma~\ref{lem:bitblastExpr}.
The other case is that of composition, in which case the theorem follows from Theorem~\ref{thm:spaceConsume}.
All other cases follow from basic properties of program composition and algebraic manipulation of exponents.
\end{proof}

\section{Discussion}\label{sec:concl}

Our work touches the areas of
formalised computability and complexity theory,
algorithm refinements and program synthesis,
and certified compilation.

Notable work in the first area includes Norrish's~\cite{norrishComputabilityThy}
formalisation results, like Rice's theorem, in a lambda calculus model.
Xu et al.~\cite{urbanTuringMachines} and Asperti and Ricciotti~\cite{turingMachinesMatita}
proved basic results on Turing machines,
using Hoare-like logics to reason about deeply-embedded machines.
Carneiro~\cite{carneiropartrec} proved basic results using partial recursive functions.
The richest line of work is by Forster~\cite{forsterThesis},
modelling the call-by-value lambda calculus, Turing machines, and other models, and proving simulations between them.
Forster and Kunze~\cite{coqexctractionlambda} introduced a largely automated synthesis framework
from native Coq functions to said lambda calculus model,
including semi-automated time bound proofs.
This was used for serious algorithms, like the Cook-Levin theorem~\cite{CookLevinCoq}.
However, their model is not reasonable in space in general.
Our contribution is the first synthesis method into a simple computation model that is reasonable in time and space.

To make our method practical, we used ideas from algorithm refinements and transport-based program synthesis~\cite{transport,transportcoq,huffmanTransfer},
introducing a novel synthesis for recursive functions in \cref{sec:natify_funs}.
Other related approaches are algorithm refinements, which were deeply studied by Lammich~\cite{LammichRefinement,lammichSepRef,lammichMonads,LammichFlows}.

In the area of certified compilation, most substantial developments geared towards practical compilers,
such as CakeML~\cite{ProofProducingML} and CompCert~\cite{compcert}.
We, on the other hand,
we were mainly inspired by the simple and elegant use of $\imp$ languages~\cite{winskelOpSem}
in educational material~\cite{coqSemanticsIMP,nipkowSemantics},
which is sufficient to formalise the theory of efficient algorithms and complexity theory.

\paragraph{Future Work}
Our framework can be extended in several ways,
such as proving time and memory bounds of synthesised programs with respect to some sensible semantics for \isabellehol functions~\cite{fdsBook},
generalising from tail-recursion to general recursion,
reducing the time and memory overhead of the compilers for finer analyses needed in efficient algorithms,
and extending it to richer models of computation,
e.g.\ probabilistic computation, interactive computation, or online computation.
We plan to improve the framework's performance
such that it can be applied
to the verification of serious computational objects,
e.g.\ reductions, like the Cook-Levin theorem~\cite{cooklevin} or Karp's 21 NP-complete problems,
and the verification of running time bounds of complex algorithms in general, e.g.\
flow algorithms~\cite{LammichFlows,ScalingIsabelle} or matching algorithms~\cite{blossomIsabelleFull}.
A challenge with our framework, as in any system that synthesises programs using natural numbers as the sole datatype, is the computational overhead from encoding/decoding datatypes into natural numbers.
Without special care, unsuspecting users may discrepancies between the expected and actual running time of synthesised programs.
Methods to smoothly handle that discrepancy would greatly improve the usability of our framework and are a possible future direction.

\camerareadytext{%
\subsubsection{Acknowledgements}
The authors thank
Bilel Ghorbel,
Florian Keßler,
Max Lang,
Nico Lintner,
Jay Neubrand,
Jonas Stahl, and
Andreas Vollert
for their contributions as part of their student coursework at TU Munich.
}
{}
\clearpage
\bibliographystyle{splncs04}
\bibliography{bibliography,long_paper}

\arxivsubmittext{}{\clearpage}
\appendix
\section{$\hol^{(\mathtextfont{\tcabbrev})}$ to $\hol^{(\mathtextfont{\tcabbrev})\natty}$}\label{sec:holtctoholtcnat_details}

\let\oldholtc\holtc
\let\oldholtcnat\holtcnat
\let\oldtermholtc\termholtc
\let\oldtermholtcnat\termholtcnat
\renewcommand{\holtc}{\hol}
\renewcommand{\holtcnat}{\holnat}
\renewcommand{\termholtc}{\termhol}
\renewcommand{\termholtcnat}{\termholnat}

\subsection{Compilation of Datatypes}\label{sec:natify_datatypes_details}

Let $\vec{\alphaty}\app t=C_1\app \vec{\alphaty_1}\mid\dotsb\mid C_n\app\vec{\alphaty_n}$ be an algebraic datatype
with $a_i\define \abs{\vec{\alphaty_i}}$.
Recall that we require an injective pairing function
$\pairnat\holhasty\natty^2\purefun\natty$
and functions
$\fstnat,\sndnat\holhasty\natty\purefun\natty$
with
\begin{gather}
\fstnat\app(\pairnat\app n\app m)=n\qquad\text{and}\qquad\sndnat\app(\pairnat\app n\app m)=m.
\end{gather}
Our construction proceeds as follows:
\begin{enumerate}[label=(\alph*)]
\item We define $\termholnat{C_i}\holhasty \natty^{a_i}\purefun \natty$ by
\begin{gather}
\termholnat{C_i}\app\vec{x}\define
\begin{cases}
\pairnat\app i\app (\pairnat\app x_1\app (\pairnat\app x_2(\dotsb(\pairnat\app x_{a_i-1}\app x_{a_i})\dotsb),&\text{if }a_i>0\\
\pairnat\app i\app 0,&\text{otherwise}
\end{cases}.
\end{gather}
\item
For each $i,j$, we define the \emph{selector} $\natselector{i}{j}\holhasty \natty\purefun \natty$ as
\begin{gather}
\natselector{i}{j}\app x\define \parenths{\holif{j<i}{\fstnat}{\id}}\app(\sndnat^j\app x).
\end{gather}
\item
We define
$\casenat\holhasty \parenths[\big]{\parenths{\natty^{a_1}\purefun\alphaty}\purefun{\dotsb\purefun\parenths{\natty^{a_n}\purefun\alphaty}\purefun\natty\purefun\alphaty}}$
by
\begin{gather}
\casenat\app\vec{f}\app x\define{}
\begin{cases}
f_i\app\parenths{\natselector{a_i}{1}\app x}\dotsb \parenths{\natselector{a_{i}}{a_i}\app x},&\text{if }i=\fstnat\app x, 1\leq i\leq n-1\\
f_n\app\parenths{\natselector{a_n}{1}\app x}\dotsb \parenths{\natselector{a_n}{a_n}\app x},&\text{otherwise}.
\end{cases}
\end{gather}
\item We define $\natify\holhasty \vec{\alpha}\app t\purefun\natty$
and $\denatify\holhasty \natty\purefun \vec{\alpha}\app t$
by
\begin{align}
\natify\app (C_i\app\vec{x})&\define \termholnat{C_i}\app\vec{\natify\app x_j}\label{eq:natify},\\
\denatify\app (\termholnat{C_i}\app\vec{x})&\define C_i\app\vec{\denatify\app x_j}\label{eq:denatify}.
\end{align}
\end{enumerate}
Note that $\natify,\denatify$ are overloaded and \cref{eq:natify,eq:denatify}
well-defined if all $\alpha\in \vec{\alpha_1}\cup\dotsb\cup\vec{\alpha_n}$
are already encoded or \makebox{$\alpha=\vec{\alpha}\app t$}.

\begin{theorem*}
$\vec{\alphaty}\app d=C_1\app \vec{\alphaty_1}\mid\dotsb\mid C_n\app\vec{\alphaty_n}$ is encodable
if all $\alpha\in \vec{\alpha_1}\cup\dotsb\cup\vec{\alpha_n}$ are encoded or \makebox{$\alpha=\vec{\alpha}\app d$}.
Also, for all $1\leq i\leq n$, $1\leq j\leq a_i$ and relations $R$, we have
\begin{align*}
&\parenths{\funrel{\relnat^{a_i}}{\relnat}}\app \termholnat{C_i}\app C_i,\qquad
\app\relnat\app n\app(C_i\app \vec{x})\implies \relnat\app (\natselector{a_i}{j}{n})\app x_j,\\
&\parenths[\big]{\funrel{\parenths{\funrel{\relnat^{a_1}}{R}}}{\funrel{\dotsb\relarrow\parenths{\funrel{\relnat^{a_n}}{R}}}{\funrel{\relnat}{R}}}}\app \casetypenat{d}\app \casetype{d}.
\end{align*}
\end{theorem*}
\begin{proof}
We have to show that $\denatify\app(\natify\app x)=x$ for all $x$.
This follows directly by definition of $\denatify,\natify$ and
the fact that $\fstnat,\sndnat$ are inverse to $\pairnat$ and all
$\alpha\in \vec{\alpha_1}\cup\dotsb\cup\vec{\alpha_n}$ are encoded (and hence have inverse functions $\natify,\denatify$).

The relatedness statements follow similarly by definition and the relatedness properties of the type arguments en- and decoding functions.
\end{proof}

\begin{example}\label{ex:natify_datatype}
For the type of lists, \lstinline{'a list = Nil | Cons 'a "'a list"},
the following constants are defined
\begin{lstlisting}[language=isabelle]
definition Nil_nat = pair_nat 1 0
definition Cons_nat [*x*] [*xs*] = pair_nat 2 (pair_nat [*x*] [*xs*])
definition case_list_nat [*f1*] [*f2*] [*n*] = if fst [*n*] = 1 then [*f1*]
  else [*f2*] (select$_{2,1}$ [*n*]) (select$_{2,2}$ [*n*])
definition denatify =
  case_list_nat Nil ($\lambda$[**x**] [**xs**]. Cons (denatify [**x**]) (denatify [**xs**]))
definition natify =
  case_list Nil_nat ($\lambda$[**x**] [**xs**]. Cons_nat (natify [**x**]) (natify [**xs**]))
\end{lstlisting}
and the following theorems are proven:
\begin{lstlisting}[language=isabelle]
lemma "$\relnat$ Nil_nat Nil"
lemma "($\relnat$ $\relarrow$ $\relnat$ $\relarrow$ $\relnat$) Cons_nat Cons"
lemma "(R $\relarrow$ ($\relnat$ $\relarrow$ $\relnat$ $\relarrow$ R) $\relarrow$ $\relnat$ $\relarrow$ R)
  case_list_nat case_list"
lemma "$\relnat$ [*ns*] ([*x*] # [*xs*]) $\implies$ $\relnat$ (select$_{2,1}$ [*ns*]) [*x*]"
lemma "$\relnat$ [*ns*] ([*x*] # [*xs*]) $\implies$ $\relnat$ (select$_{2,2}$ [*ns*]) [*xs*]"
\end{lstlisting}
\end{example}

\subsection{Synthesis of $\holtcnat$ Functions}\label{sec:natify_funs_details}

\begin{theorem*}
$\relnat$, $\natify$, and $\denatify$ form a partial Galois equivalence.
Thus, there is some $\termholnat{f}$ that is $\relnat$-related to $\termhol{f}$.
\end{theorem*}
\begin{proof}
Due to~\cite[Lemma 2]{transport}, it suffices to show that
\begin{enumerate*}
\item $\relnat$ is right-total and right-unique,
\item $\relnat\app n\app x$ implies $\denatify\app n = x$, and
\item $\relnat\app (\natify \app  x)\app x$
\end{enumerate*}
for all $n,x$.
By \cref{def:nat_encoded},
$\relnat$ is right-total and $\relnat\app (\natify \app  x)\app x$.
Since $\natify$ is injective ($\cref{def:nat_encoded}$),
$\relnat$ is also right-unique
and $\relnat\app n\app x$ implies $\denatify\app n = x$.

The existence of $\termholnat{f}$ now follows from~\cite[Theorem 1 and 2]{transport}.
\end{proof}

\begin{theorem*}[Whitebox-Transport~\cite{huffmanTransfer}]
$\whiteboxtransport{(\lambda\vec{x}.\app \termholtc{t})}$ and $\lambda\vec{x}.\app \termholtc{t}$ are $\relnat$-related
and white-box transports preserve tail-recursiveness and first-order applications.
\end{theorem*}
\begin{proof}
The former is shown by structural induction on $t$,
using that $\termholtc{f}$ is $\relnat$-related to $\termholtcnat{f}$ by~\cref{thm:blackbox}
and all other functions $\termholtc{g}$ are $\relnat$-related to $\termholtcnat{g}$ by assumption.

The latter follows immediately from \cref{def:typerelnatsubst} and structural induction on $t$.
\end{proof}

\renewcommand{\holtc}{\oldholtc}
\renewcommand{\holtcnat}{\oldholtcnat}
\renewcommand{\termholtc}{\oldtermholtc}
\renewcommand{\termholtcnat}{\oldtermholtcnat}

\section{$\foltcnat$ to $\imptc$}\label{sec:holtcnattoimptc_details}

\subsection{Compilation to $\imptc$}\label{sec:holtcnattoimptc_compiler_details}

\cref{fig:holtcnat_to_imp_datatype} shows the datatype representation of $\foltcnat$ as used in the compiler
to $\imptc$, shown in \cref{fig:holtcnat_compiler_appendix}.
\begin{figure}[t]
\begin{gather*}
\begin{aligned}
t \Coloneqq  & \; \holnIf{t_1}{t_2}{t_3} & \text{(if-then-else with natural number condition $t_1$)}\\
  \mid & \; \holnLet{t_1}{t_2} \quad & \text{(bind $t_1$ to the first de Bruijn index in $t_2$)} \\
  \mid & \; \holnLetBound{i} \quad & \text{(variable bound by $i$-th enclosing \textbf{Let})} \\
  \mid & \; \holnArg{i} \quad & \text{($i$-th argument of $f$, i.e.\ $x_i$)} \\
  \mid & \; \holnNum{n} & \text{(natural number $n \in \natty$)} \\
  \mid & \; \holnCall{g}{t_1, \dotsc, t_m} \quad & \text{(call to $g \holhasty \natty^m \purefun \natty$ with arguments $t_1, \dotsc, t_m$)} \\
  \mid & \; \holnTailCall{t_1, \dotsc, t_k} \quad & \text{(recursion with arguments $t_1, \dotsc, t_k$)}
\end{aligned}
\end{gather*}
\caption{Datatype representation of $\foltcnat$ in the metaprogramming language.}\label{fig:holtcnat_to_imp_datatype}
\end{figure}
\begin{figure}[t]
\begin{gather*}
\begin{aligned}
\dens{\holnIf{t_1}{t_2}{t_3}}{b}{r} \define &\; \dens{t_1}{b}{x} \seqi \iif{x}{\dens{t_2}{b}{r}}{\dens{t_3}{b}{r}} & \text{(fresh $x$)} \\
\dens{\holnLet{t_1}{t_2}}{b}{r} \define &\; \dens{t_1}{b}{x} \seqi \dens{t_2}{x \# b}{r} & \text{(fresh $x$)} \\
\dens{\holnLetBound{n}}{b}{r} \define &\; \assign{r}{b\app {!}\app n} & \\
\dens{\holnArg{n}}{b}{r} \define &\; \assign{r}{\htoiargi{f}{n}} & \\
\dens{\holnNum{n}}{b}{r} \define &\; \assign{r}{n} & \\
\dens{\holnCall{g}{t_1, \dotsc, t_m}}{b}{r} \define &\; \dens{t_1}{b}{x_1} \seqi {\dotsc} \seqi \dens{t_m}{b}{x_m} \seqi & \text{(fresh $x_1, \dotsc, x_m$)}  \\
    &\; \assign{\htoiargi{g}{1}}{x_1} \seqi \dotsc \seqi \assign{\htoiargi{g}{m}}{x_m} \seqi & \\
    &\; \call{g^{\imp}}{\htoiret{g}} \seqi  \assign{r}{\htoiret{g}} &\text{($g^{\imp}$ registered for $g$)}\\
\dens{\holnTailCall{t_1, \dotsc, t_k}}{b}{r} \define &\; \dens{t_1}{b}{x_1} \seqi {\dotsc} \seqi \dens{t_k}{b}{x_k} & \text{(fresh $x_1, \dotsc, x_k$)} \\
    &\; \assign{\htoiargi{f}{1}}{x_1} \seqi \dotsc \seqi \assign{\htoiargi{f}{k}}{x_k} \seqi & \\
    &\; \tail &
\end{aligned}
\end{gather*}
\caption{The compiler from the $\foltcnat$ representation in \cref{fig:holtcnat_to_imp_datatype} to $\imptc$.}\label{fig:holtcnat_compiler_appendix}
\end{figure}

\subsection{Correctness Proofs}\label{sec:holtcnattoimptc_correctness_details}
Our goal is to show that $\bigstepCtxtPart{p}{p}{s}{\htoiret{f}}{\termfoltcnat{f}\app \vec{\stateretr{s}{\htoiargi{f}{i}}}}$,
i.e.\ we are interested in the single return register $\htoiret{f}$.
\cref{fig:exec_res_imptc} shows the execution relation for $\imptc$ specialised to single return registers.
\begin{figure}[t]
\begin{mathpar}
  \inferrule*[left=\assignN$_1$]
  {\stateretr{\funupd{s}{r}{\aval{a}{s}}}{r'}=v}
  {\bigstepPart{\assign{r}{a}}{s}{r'}{v}}[\assignN$_1$]

  \inferrule*[left=\ifTrueN$_1$]
  {\stateretr{s}{r} \neq 0 \\ \bigstepPart{p_1}{s}{r}{v}}
  {\bigstepPart{\iif{r}{p_1}{p_2}}{s}{r}{v}}[\ifTrueN]

  \inferrule*[left=\seqN$_1$]
  {\bigstep{p_1}{s}{s'} \\
  \bigstepPart{p_2}{s'}{r}{v}}
  {\bigstepPart{\seq{p_1}{p_2}}{s}{r}{v}}[\seqN$_1$]\label{infrule:exec_res_imptc_seq}

  \inferrule*[left=\ifFalseN$_1$]
  {\stateretr{s}{r} = 0 \\ \bigstepPart{p_2}{s}{r}{v}}
  {\bigstepPart{\iif{r}{p_1}{p_2}}{s}{r}{v}}[\ifFalseN$_1$]

  \inferrule*[left=\callN$_1$]
  {\bigstepPart{pc}{s}{r}{v} \\
  \stateretr{\funupd{s}{r}{v}}{r'} = v'}
  {\bigstepPart{\call{pc}{r}}{s}{r'}{v'}}[\callN$_1$]

  \inferrule*[left=\tailN$_1$]
  {\bigstepCtxtPart{p}{p}{s}{r}{v}}
  {\bigstepCtxtPart{p}{\tail}{s}{r}{v}}[\tailN$_1$]
\end{mathpar}
\caption{Execution relation of $\imptc$ for single return registers.}\label{fig:exec_res_imptc}
\end{figure}
Note that for sequences \makebox{$\bigstepPart{\seq{p_1}{p_2}}{s}{r}{v}$},
the standard execution relation ``$\bigsteparrow$'' (and not ``$\bigsteparrow_r$'') must be used for $p_1$,
that is $\bigstep{p_1}{s}{s'}$.

This is the reason why each compiled $\imptc$ program $p$ is normalised such that no recursive constructor (sequences and if-then-elses)
appears on the left of a sequence.
This way, all goals involving recursive constructors use the execution relation ``$\bigsteparrow_r$'',
saving us from implementing separate recursive automation for~``$\bigsteparrow$''.

\section{$\imptc$ to $\impminus$}

\subsection{$\imptc$ to $\impc$}\label{sec:imptctoimpc_details}
\begin{lemma*}
\begin{figure}[t]
\begin{mathpar}
  \inferrule*
  {s'=\funupd{s}{r}{\aval{a}{s}}}
  {\bigstepTime{\assign{r}{a}}{s}{\const}{(s',\false)}}

  \inferrule*
  {\stateretr{s}{r} \neq 0 \\ \bigstepTime{p_1}{s}{n}{(s',b)}}
  {\bigstepTime{\iif{r}{p_1}{p_2}}{s}{n+\const}{(s',b)}}

  \inferrule*
  {\bigstepTime{p_1}{s}{n_1}{(s', \false)}\\
  \bigstepTime{p_2}{s'}{n_2}{(s'', b)}}
  {\bigstepTime{\seq{p_1}{p_2}}{s}{n_1+n_2+\const}{(s'',b)}}

  \inferrule*
  {\stateretr{s}{r} = 0 \\ \bigstepTime{p_2}{s}{n}{(s',b)}}
  {\bigstepTime{\iif{r}{p_1}{p_2}}{s}{n+\const}{(s',b)}}

  \inferrule*
  {\bigstepTimePart{pc}{s}{n}{r}{v} \\
   s'=\funupd{s}{r}{v}}
  {\bigstepTime{\call{pc}{r}}{s}{n+\const}{(s',\false)}}

  \inferrule*
  {$ $}
  {\bigstepTime{\tail}{s}{\const}{(s,\true)}}
\end{mathpar}
\caption{Auxiliary execution relation used in \cref{lem:tailWhileRules}.}\label{fig:exectailcallflag_complete}
\end{figure}

Let $\bigstepTime{p}{s}{n}{(s',b)}$
denote the modified execution relation shown in \cref{fig:exectailcallflag_complete} that runs
$p$ but stops and flags whenever a tail-recursive call is encountered using \makebox{$b\in\{\false,\true\}$}
Then the following rules are admissible:
\begin{mathpar}
\inferrule*
  {\bigstepTime{p}{s}{n}{(s',\false)}}
  {\bigstepTime{p}{s}{n}{s'}}

\inferrule*
  {\bigstepTime{p}{s_1}{n_1}{(s_2,\true)}\\
   \bigstepCtxtTime{tp}{tp}{s_2}{n_2}{s_3}}
  {\bigstepCtxtTime{tp}{p}{s_1}{n_1+n_2}{s_3}}
\end{mathpar}
\end{lemma*}
\begin{proof}[Proof sketch]
By induction over the execution rules.
\end{proof}

\begin{theorem*}
Let $\impconstrfont{cnt}$ be a fresh register with respect to $\vars\app p$.
Let $p\psubst{t}{x}$ denote the syntactic substitution of any occurrence of $x$ by $t$ in a program $p$.
Define
\begin{gather*}
\imptctoimpc{p} \define
\seq{\assign{\impconstrfont{cnt}}{\impconstrfont{1}}}{\while{\impconstrfont{cnt}}{\seq{\assign{\impconstrfont{cnt}}{\impconstrfont{0}}}{p\psubst{(\assign{\impconstrfont{cnt}}{\impconstrfont{1}})}{\tail}}}}.
\end{gather*}
Then it holds that
$\bigstepCtxtTimePart{p}{p}{s}{n}{\vars{p}}{s'}$ \textiff $\bigstepTimePart{\imptctoimpc{p}}{s}{n+\const}{\vars{p}}{s'}$.
\end{theorem*}
\begin{proof}[Proof sketch]
We show that
$\bigstepDelimTime{p}{s}{n}{s'}{b}$ holds
\textiff
\[
\exists s''.\app \bigstepTime{p\psubst{(\assign{\impconstrfont{cnt}}{\impconstrfont{1}})}{\tail}}{s}{n}{s''} \land s''\app \impconstrfont{cnt} = b \land\forall r\in\vars{p}.\app s'\app p = s''\app p
\]
by induction over the execution rules, generalising over $S$
such that $\vars{p}\subseteq S$ and $\impconstrfont{cnt}\notin S$.
Then the semantics for the while-loop with rules from \cref{fig:whilerules} are equivalent to the semantics of the compiled program with rules from \cref{lem:tailWhileRules} by straightforward induction over the respective rules.
\end{proof}

\subsection{$\impc$ to $\impw$}\label{sec:impctoimpw_details}

\begin{lemma*}
$\bigstepTime{p\psubst{m\app x}{x}}{s}{n}{s'}$ implies $\bigstepTime{p}{s\comp m}{n}{s'\comp m}$
for injective $m$.
\end{lemma*}
\begin{proof}[Proof sketch]
By induction over the execution of $p$.
\end{proof}
The inlining induces a blow-up that is dependent on the program size (but not the input), as shown below.
\begin{theorem*}
$\bigstepTimePart{p}{s}{n}{\vars\app p}{s'}$ \textiff $\bigstepTimePart{\impctoimpw{p}}{s}{f(n)}{\vars\app p}{s'}$,
where $f(n)$ is linear in $n \cdot \abs{p}$.
\end{theorem*}
\begin{proof}[Proof sketch]
The inlining is correct by \cref{lem:neat_subst}, so correctness follows by a simple induction.
\end{proof}

\subsection{$\impw$ to $\impminus$}\label{sec:impwtoimpminus_details}

\begin{lemma*}
Assume that $\maxset\{\maxconst{a}, \aval{a}{s}, \maxconst{s}\}< 2^w$.
Then it holds that\newline $\bigstepTime{\bitblastAexp{a}{w}{r}}{\bitblastState{s}{w}}{f(w)}{\bitblastState{\funupd{s}{r}{\aval{a}{s}}}{w}}$,
where $f(w)$ is linear in $w$.
\end{lemma*}
\begin{proof}[Proof Sketch]
The proof is based on showing an upper bound on the running time of each of the functions $\copyopnd{b}{y}{a}$ and $\adder{v}{w}$, which in turn depends on the running time of $\fulladder{r}{i}$.
The running time of $\fulladder{r}{i}$ is a constant.
The running times of $\adder{}{}$ and $\copyopnd{}{}{}$ are shown to be linear in $w$ by induction on~$w$.
\end{proof}

\section{$\imp$-Definitions}\label{sec:all_imps}
\cref{fig:imptc_complete} shows the complete specification of $\imptc$,
\cref{fig:whilerules_appendix} the rule for $\whileop$, and
\cref{fig:imp_minus_def_appendix} the complete specification of $\impminus$.

\begin{figure}[t]
\begin{subfigure}[t]{0.4\textwidth}
\begin{gather*}
\aval{\aconst{n}}{s}\define n,
\hspace{0.9em}
\aval{\avar{r}}{s}\define \stateretr{s}{r}
\end{gather*}
\caption{Evaluation of atoms $A$.}
\end{subfigure}
\begin{subfigure}[t]{0.6\textwidth}
\begin{gather*}
\aval{A_1\otimes A_2}{s}\define \aval{A_1}{s} \otimes \aval{A_2}{s},
\text{ for } \otimes\in\{+,-\}
\end{gather*}
\caption{Evaluation rules for arithmetic expressions.}
\end{subfigure}
\begin{subfigure}[t]{\textwidth}
\begin{mathpar}
  \inferrule*[left=\assignN]
  {s'=\funupd{s}{r}{\aval{a}{s}}}
  {\bigstepTime{\assign{r}{a}}{s}{\const}{s'}}[\assignN]

  \inferrule*[left=\ifTrueN]
  {\stateretr{s}{r} \neq 0 \\ \bigstepTime{p_1}{s}{n}{s'}}
  {\bigstepTime{\iif{r}{p_1}{p_2}}{s}{n+\const}{s'}}[\ifTrueN]

  \inferrule*[left=\seqN]
  {\bigstepTime{p_1}{s}{n_1}{s'} \\
  \bigstepTime{p_2}{s'}{n_2}{s''}}
  {\bigstepTime{\seq{p_1}{p_2}}{s}{n_1+n_2+\const}{s''}}[\seqN]

  \inferrule*[left=\ifFalseN]
  {\stateretr{s}{r} = 0 \\ \bigstepTime{p_2}{s}{n}{s'}}
  {\bigstepTime{\iif{r}{p_1}{p_2}}{s}{n+\const}{s'}}[\ifFalseN]

  \inferrule*[left=\callN]
  {\bigstepTimePart{pc}{s}{n}{r}{v} \\
   s'=\funupd{s}{r}{v}}
  {\bigstepTime{\call{pc}{r}}{s}{n}{s'}}[\callN]

  \inferrule*[left=\tailN]
  {\bigstepCtxtTime{p}{p}{s}{n}{s'}}
  {\bigstepCtxtTime{p}{\tail}{s}{n+\const}{s'}}[\tailN]
\end{mathpar}
\caption{Execution relation for commands.}
\end{subfigure}
\caption{Semantics of $\imptc$.}\label{fig:imptc_complete}
\end{figure}

\begin{figure}[t]
  \begin{mathpar}
    \inferrule*[left=\whileFalseN]
    {\stateretr{s}{r} = 0}
    {\bigstepTime{\while{r}{p}}{s}{\const}{s}}
    [\whileFalseN]

  \inferrule*[left=\whileTrueN]
    {\stateretr{s_1}{r} \neq 0\\
     \bigstepTime{p}{s_1}{n_1}{s_2}\\
     \bigstepTime{\while{r}{p}}{s_2}{n_2}{s_3}}
    {\bigstepTime{\while{r}{p}}{s_1}{n_1+n_2+\const}{s_3}}
    [\whileTrueN]
  \end{mathpar}
  \caption{Execution relation for $\whileop$ from $\impwc$ and $\impw$.}\label{fig:whilerules_appendix}
\end{figure}

\begin{figure}[t]
  \begin{mathpar}
    \inferrule*
    {b\in\{\false,\true\}\\
    s'=\funupd{s}{r}{b}}
    {\bigstepTime{\assign{r}{b}}{s}{1}{s'}}

    \inferrule*
    {\stateretr{s}{r} \neq \false \\ \bigstepTime{p_1}{s}{n}{s'}}
    {\bigstepTime{\iif{r}{p_1}{p_2}}{s}{n+1}{s'}}

    \inferrule*
    {\bigstepTime{p_1}{s}{n_1}{s'} \\
    \bigstepTime{p_2}{s'}{n_2}{s''}}
    {\bigstepTime{\seq{p_1}{p_2}}{s}{n_1+n_2+1}{s''}}

    \inferrule*
    {\stateretr{s}{r} = \false \\ \bigstepTime{p_2}{s}{n}{s'}}
    {\bigstepTime{\iif{r}{p_1}{p_2}}{s}{n+1}{s'}}

    \inferrule*
    {\stateretr{s}{r} = \false}
    {\bigstepTime{\while{r}{p}}{s}{1}{s}}\hspace{1em}
    \inferrule*
    {\stateretr{s_1}{r} \neq \false\\
    \bigstepTime{p}{s_1}{n_1}{s_2}\\
    \bigstepTime{\while{r}{p}}{s_2}{n_2}{s_3}}
    {\bigstepTime{\while{r}{p}}{s_1}{n_1+n_2+2}{s_3}}
  \end{mathpar}
  \caption{Execution relation for $\impminus$.}\label{fig:imp_minus_def_appendix}
\end{figure}

\end{document}